\documentclass[11pt]{article}

\usepackage{graphicx}
\usepackage{epstopdf}
\usepackage{amsmath,amssymb,amsthm}
\allowdisplaybreaks
\usepackage{float}
\usepackage{mathtools}
\usepackage{nicefrac}
\usepackage{slashed}
\usepackage{empheq}
\usepackage{xcolor}
\usepackage[shortlabels]{enumitem}
\usepackage{comment}

\usepackage[left=1in, right=1in, top=1in, bottom=1in]{geometry}

\newtheorem{theorem}{\sc Theorem}[section]

\newtheorem{proposition}[theorem]{\sc Proposition}
\newtheorem{prop}[theorem]{\sc Proposition}

\numberwithin{equation}{section}
\theoremstyle{remark}
\newtheorem{remark}[theorem]{Remark}
\newtheorem{example}[theorem]{Example}
\newcommand{\be}{\begin{equation}}
\newcommand{\ee}{\end{equation}}

\providecommand{\abs}[1]{\vert#1\vert}

\def\bN{\mathbb{N}}

\def\bR{\mathbb{R}}
\def\bZ{\mathbb{Z}}

 \def\Z{\bZ} 
 
\def\R{\bR}
\def\N{\bN}

\definecolor{darkgreen}{rgb}{0.0,0.5,0.0}
\definecolor{darkblue}{rgb}{0.0,0.0,0.3}
\definecolor{nicosred}{rgb}{0.65,0.1,0.1}
\definecolor{light-gray}{gray}{0.7}

\usepackage{todonotes} % IZKISS

\title{Decoding how higher-order network interactions shape complex contagion dynamics}
\author{Istv\'an Z. Kiss$^{1,*}$, Christian Bick$^{2, 3, 4,5}$ \& P\'eter L. Simon$^{6,7}$}
\date{
$^1$Network Science Institute, Northeastern University London, London, United Kingdom\\
$^2$Department of Mathematics, Vrije Universiteit Amsterdam, Amsterdam, the Netherlands\\
$^3$Institute for Advanced Study, Technical University of Munich, Garching, Germany \\
$^4$Department of Mathematics, University of Exeter, Exeter, United Kingdom\\
$^5$Mathematical Institute, University of Oxford, Oxford, United Kingdom\\
$^6$Institute of Mathematics, E\"otv\"os Lor\'and University, Budapest, Hungary\\
$^7$Alfr\'ed R\'enyi Institute of Mathematics, HUN-REN, Budapest, Hungary\\
	  \vspace{1cm}\today
	  }

\begin{document}
%%%%%%%%%%%%%%%%%%%%%%%%%%%%%%%%%%%%%%%%%%%%%%%%%%%%%%%%%%%%%%%%	
\maketitle	
%%%%%%%%%%%%%%%%%%%%%%%%%%%%%%%%%%%%%%%%%%%%%%%%%%%%%%%%%%%%%%%%

%%%%%%%%%%%%%%%%%%%%%%%%%%%%%%%%%%%%%%%%%%%%%%%%%%%%%%%%%%%%%%%
\begin{abstract} 
Complex contagion models that involve contagion along higher-order structures, such as simplicial complexes and hypergraphs, yield new classes of mean-field models. 
% Specifically, higher-order effects appear in the mean-field equations in the form of higher-degree terms beyond the quadratic pairwise transmission in classical spreading models. 
Interestingly, the differential equations arising from many such models often exhibit a similar form, resulting in qualitatively comparable global bifurcation patterns. Motivated by this observation, we investigate a generalized mean-field-type model that provides a unified framework for analysing a range of different models.
% \IZK{Mention that many existing model have the same mathematical structure?}
In particular, we derive analytical conditions for the emergence of different bifurcation regimes exhibited by three models of increasing complexity---ranging from three- and four-body interactions to two connected populations with both pairwise and three-body interactions.
For the first two cases, we give a complete characterisation of all possible outcomes, along with the corresponding conditions on network and epidemic parameters. 
In the third case, we demonstrate that multistability is possible despite only three-body interactions. 
Our results reveal that single population models with three-body interactions can only exhibit simple transcritical transitions or bistability, whereas with four-body interactions multistability with two distinct endemic steady states is possible. 
Surprisingly, the two-population model exhibits multistability via symmetry breaking despite three-body interactions only. 
Our work sheds light on the relationship between equation structure and model behaviour and makes the first step towards elucidating mechanisms by which different system behaviours arise, and how network and dynamic properties facilitate or hinder outcomes.
\end{abstract}

%\tableofcontents

%\onecolumngrid

\newpage

%%%%
\section{Introduction}
% \IZK{Check figure captions}
% \IZK{Mean-field if Adjective is the mean field (as subject) Figure}
% \CB{I started doing some ``replace all''. For the former I use ``mean-field equation'' but ``the mean field is ...''}
% \CB{Should edit for consistency:  three-body, four-body, non-pairwise or multi-body  Agreed ...}
% \IZK{Yes, it is a mess with simplicial, 2-body, 3-body etc}
% \IZK{Also we have very few references in my opinion, maybe it is ok.}
% \CB{Re ``triplet'' and co: I will try to do some clean up when I go through again. Re references, agreed, would be good to weave in a few more.}
The development and analysis of contagion models on higher-order networks over the last few years opened up new directions of research to understand spreading processes on networks; cf.~\cite{battiston2020networks,Bick2021,ferraz2024contagion}.
Classical epidemic models on networks where contagion acts along an edge~$\{i,j\}$ between nodes~$i$ and~$j$.
By contrast, the probability of a susceptible node~$i$ becoming infected in a network with higher-order interactions depends on the order~$\ell$ of a simplex/hyperedge~$\{i, j_1,\dotsc j_{\ell-1}\}$ that represents a $\ell$-body interaction involving~$i$.
For example in case of simplicial complexes---as illustrated in Figure~\ref{fig:schematic_simplicial_infection}---the rate of infection of susceptible node~$i$ is significantly increased by routes of infection via 2- and 3-simplices.
Note that complex contagion induced by simplices/hyperedges is in addition to the pairwise contagion dynamics. 
While three distinct nodes~$i$,~$j$ and~$k$ may be connected in a triangle of pairwise links, the existence of~$\{i,j,k\}$ indicates complex contagion. 
If $\{i,j,k\}$ is part of a simplicial complex (i.e., it satisfies the closure relation that $\{i,j\}$, $\{j,k\}$, $\{i,k\}$ are also present), the complex contagion implies an additional infection pressure onto a susceptible node whose neighbours are infected, see Figure~\ref{fig:schematic_simplicial_infection}. 
More generally for hypergraphs, each hyperedge exerts infection pressure that typically takes the form~$\nu I^{\alpha}$ with~$I$ being the number of infected nodes in the hyperdege and parameters~$\alpha,\nu$~\cite{st-onge_social_2021,st-onge_influential_2022,de_arruda_social_2020}. 
Note that this allows for any number of infected nodes within the hyperedge in contrast to simplicial contagion where all nodes in the simplex apart from the susceptible node have to be infectious for the excess contagion to be activated.

Complex contagion through higher-order interactions yields dynamical phenomena one may not expect if contagion is only pairwise. 
In contrast to traditional pairwise contagion, which is typically characterized by the disease free state losing stability to an endemic state in a forward transcritical bifurcation, contagion on higher-order networks can show hysteresis/bi-stability.
For example, the disease-free and an endemic equilibrium may co-exist or there could be multistability between two endemic equilibria (strictly positive steady states).
For example in~\cite{iacopini2019simplicial}, the authors show analytically that the transition is discontinuous and that a bistable region appears where healthy and endemic states co-exist.
Similarly, in~\cite{ghosh2023dimension}, the authors construct a one-dimensional model starting from an ensemble model based on the degree of the nodes. 
After a number of averaging assumption, their one dimensional mean-field model also displays discontinuous transition and bi-stability. 
Even more complex behaviour is observed for contagion dynamics on hypergraphs. In~\cite{ferraz2023multistability}, the authors show that their model has a vast parameter space, with first- and second-order transitions, bistability, and hysteresis. Finally, the authors of~\cite{kiss_insights_2023} show rigorously that the limiting mean field model of an exact complex contagion model with arbitrary simplicial complexes on a fully connected network reduces to a single differential equation which displays bistability and multi-stability with two strictly positive endemic steady state co-existing.

\begin{figure}%[h]
     \centering
     \includegraphics[scale=0.6]{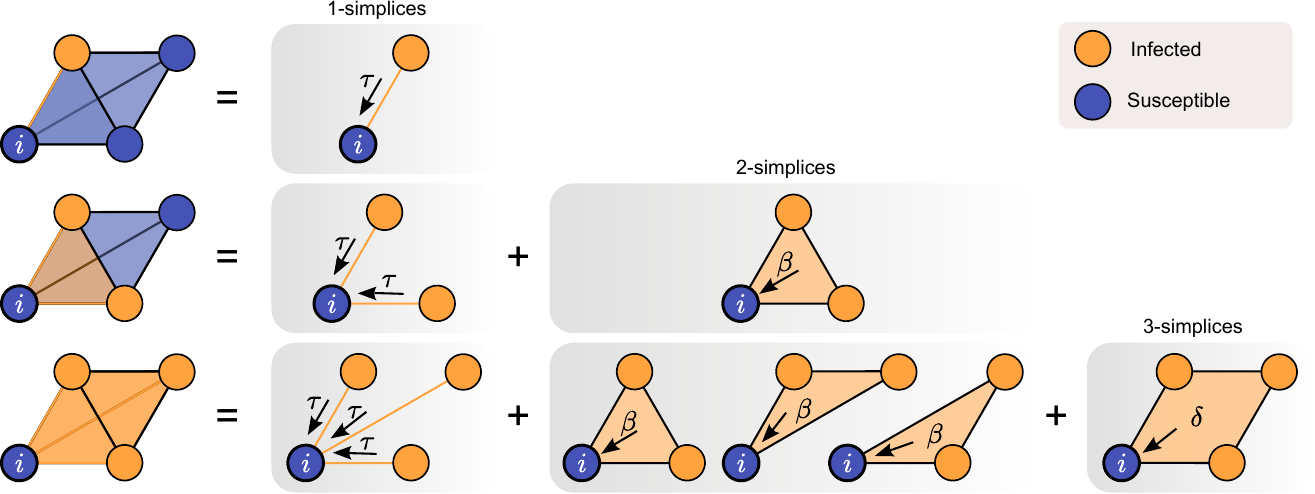}
     \caption{
     Illustration of the different possible routes of infection when higher-order interactions are modelled as simplicial complexes. The susceptible~(S) ~$i$ is part of a 3-simplex (left column); being part of all resulting lower-order simplices. The infection pressure that~$i$ receives grows with the number of infectious nodes~(I) within the
considered simplex. 
Top row: One node is infectious, thus~$i$ can be infected by a single 1-simplex (edge) with rate $\tau > 0$. 
Middle row: Two nodes are infectious, thus~$i$ can be infected either by the two 1-simplices or by the ``infectious'' 2-simplex (triangle) with rate $\beta> 0$. 
Bottom row: All three nodes are infectious, thus~$i$ can be infected either by the three 1-simplices, by the three 2-simplices, or by the 3-simplex (tetrahedron) with rate $\delta>0$.
% \textcolor{red}{Figure placement is weird.}{\color{green}{Place it as you see fit! Maybe the top of next page? JOurnal will do it anyway.}}
}
     % Epidemic dynamics in higher-order networks exhibiting 2- and 3-simplicies. Traditional infection through a link happens at rate $\tau$, while the rate of infection due to a 2-simplex is $\beta$, provided that the susceptible node has two infectious neighbours. Finally, infection via a 3-simplex is at rate $\delta$ provided that a susceptible nodes has three infectious neighbours. Note that infection within a simplex encompass all resulting lower order infections, see bottom row.}
\label{fig:schematic_simplicial_infection}
\end{figure}

{From a mathematical viewpoint, modelling of complex contagion on general structures can be done either (a) considering the entire system at once, referred to as top-down, or (b) starting at the node-level and building up to pairs, triples and so on, referred to as bottom-up~\cite{kiss2017mathematics}.}
For the top-down approach, the state of microscopic dynamics of susceptible/infected/susceptible (SIS) contagion model at the individual level is given by an $N$-dimensional vector of zeros and ones~\cite{bodo2016sis, kiss_insights_2023}, leading to a Markov chain on a state space with $2^N$~elements. 
Another widely used approach is the bottom-up approach where one starts with evolution equations for the states of the nodes. These of course depend on the joint probabilities of pairs, the states of the pairs on the other hand depend on the states of the triples and so on. This induces a hierarchy of dependencies with the number of required equations becoming difficult to manage and analyse. In a desire to break these dependencies, higher-order moments are approximated by node-level quantities  leading to so-called individual-based mean-field models~\cite{kiss2017mathematics, van2008virus}.
To allow for analytical and/or numerical treatment, both the exact bottom-up and top-down models require closures. However, the bottom-up model has some advantages in that the contact structure is evident in the equations and the order at which closures are performed and the closure relationships themselves are more easily accessible. In general, closures lead to low dimensional systems of differential equations where bifurcation analysis can be employed to reveal possible behaviours and their corresponding parameter ranges.

When simplified further to allow for an explicit analysis---often to a single differential equation assuming homogeneity---many of these models show qualitatively similar behaviour at the system level.
Indeed, the resulting equations are given by polynomial vector field whose polynomial order relates to the interaction order and exhibit comparable global bifurcation structure.
%However, many of the complex contagion models at the individual-level (with closure applied), which based on further homogeneity arguments about the contact structure are reduced to one single differential equation, show striking similarity in the form of the equation (i.e., polynomial with its order related to the level of the interaction) and the qualitative behaviour at system level, i.e., global bifurcation picture.
%To support the statement above, at least from the viewpoint of equations similarity, three models are given explicitly below. 
We give three concrete examples that describe the evolution of the fraction of infected individuals~$y(t)$.
In~\cite{iacopini2019simplicial}, the authors proposed the following model
\begin{align}
  \textbf{Model A:}\qquad  \dot{y}(t)&=-\mu y(t)+\sum_{\omega=1}^L\beta_\omega\langle k_\omega\rangle y^\omega(t)[1-y(t)],
    % &=(\beta_{1}\langle k_{1}\rangle-\mu)\rho(t)+\sum_{\omega=1}^{D-1}\left(\beta_{\omega+1}\langle k_{\omega+1}\rangle-\beta_{\omega}\langle k_{\omega}\rangle\right)\rho^{\omega+1}(t)-\beta_{D}\langle k_{D}\rangle \rho^{D+1}(t),
    \label{eq:ex_iacopo}
\intertext{
where~$\mu$ is the rate of recovery, $\beta_{\omega}$ is simplicial-order specific rate of infection, and $\langle k_\omega\rangle$ is the average number of order~$\omega$ simplices a typical node is part of. Starting from the level of nodes and taking into account degrees, the authors of~\cite{ghosh_dimension_2023} use a series of approximations to arrive at}
   \textbf{Model B:}\qquad \dot{y}(t)&=-y(t)+\lambda\beta_{\mathrm{eff}}(1-y(t))y(t)+\gamma\beta_{\mathrm{eff}}(1-y(t))y^{2}(t),
    % &=(\lambda \beta_{\mathrm{eff}}-1)x_{\mathrm{eff}}+(\gamma \beta_{\mathrm{eff}}-\lambda \beta_{\mathrm{eff}})x_{\mathrm{eff}}^2-\gamma \beta_{\mathrm{eff}}x_{\mathrm{eff}}^3,
     \label{eq:ex_degree}
\intertext{
which can be seen as the model above with $\omega=2$.
Finally, in~\cite{kiss_insights_2023}, the authors start from a fully connected network and model contagion over simplices of arbitrary order to obtain 
% \begin{equation}
%      \dot{y}(t)=(s_{1}-\gamma) y(t)+\sum_{i=1}^{M-1}\left(\frac{s_{i+1}}{(i+1)!}-\frac{s_{i}}{i!}\right)y^{i+1}(t)-\frac{s_{M}}{M!}y^{M+1}(t),
% \end{equation}
}
     \textbf{Model C:}\qquad\dot{y}(t)&=-\gamma y(t)+\sum_{i=1}^{M}\frac{s_{i}}{i!}\left(1-y(t)\right)y^{i}(t).
      \label{eq:ex_kiss_insights}
\end{align}
From the models above, it is evident that all  right hand sides are polynomials where the coefficients encode the underlying structure and the complex contagion mechanism.

Motivated by these observations, in this paper, we provide general classification results for the possible dynamical behavior of a wide class of complex contagion dynamics.
Specifically, we focus on a class of vector fields that encompass many different models including Models~A,~B, and~C above.
This approach allows us to elucidate commonality between models and to map out explicitly the dependence between equation structure and model outcome.
We highlight how model behaviour depends on contact network parameters including pairwise and higher-order interactions that contribute to complex contagion dynamics.
We start from the simplest possible one-population model, in the sense that all nodes are topologically equivalent (meaning that on average they partake in the same number of links, 2-simplices, 3-simplices, hyperedges of various sizes and so on).
This allows us to give a full mathematical characterisation of possible outcomes and outline whether critical bifurcation points depend on higher-order structure.
We then move to two-population models where nodes are of two types; that is in terms of how they connect within their own population and with the other. 
The main findings of the paper are the identification of two distinct mechanisms of how bi- and multi-stability arises as well as linking system-level behaviour to constraints on the network structure.

The paper is structured as follows. 
In Section~\ref{sec:model}, we introduce the general class of equations for SIS dynamics on hypergraphs and simplicial complexes and outline how they can be derived from individual-based models.
Section~\ref{sec:results} focuses on single population models that can be reduced to a single dynamical equation and identifies what interaction order is necessary for bistability.
In Section~\ref{sec:TwoPop} we consider two-population models for symmetric and near-symmetric coupling.
Here multistability can already arise in the presence of triplet interactions.
We conclude with a discussion and identify potential future work in Section~\ref{sec:discussion}.

% There is no d\textbf{}oubt that this remains a fast evolving research area with many new and relpertinent questions such as (i) what 

% \textbf{General statements/question about bifurcations in the individual based  model}
% \IZK{What shall we do with these statements, where should we place these?}
% \begin{itemize}

% \item {\bf 2D case - Does the previous statement remain true for non-regular networks? This could be investigated in a special case, see the example below in \eqref{individual based _2groups}.}

% \item {\bf This is the motivation for using HYpergraphs with more degrees of freedome: Is there a bifurcation diagram that occurs for a hypergraph but does not occur for simplicial complexes?}

% \end{itemize}

% {\bf Potential journals}
% \begin{itemize}
%     \item Interface :(
%     \item Journal of Complex Networks
%     \item Journal of Math Biol
%     \item Bull of Math Biol
%     \item Siam applied maths
%     \item Siam dynamical system
%     \item Phys Rev E?
%     \item IOP complexity?
%     \item Nature :)
% \end{itemize}

\section{Models}\label{sec:model}

Consider a set of $N$~nodes that we simply enumerate by $V = \{ 1,2, \dotsc , N\}$. A \emph{hyperedge} is specified by a subset~$H\subset V$; 
the number of distinct nodes $\ell=\abs{H}$ is the order of a hyperedge.
The set of nodes together with a set~$\mathcal{H}$ of hyperedges form a \emph{hypergraph}.
Each $\ell$-uniform subhypergraph $\mathcal{H}^{(\ell)} = \{H\in\mathcal{H}: \abs{H}=\ell\}$ that contains the edges of order~$\ell$ in~$\mathcal{H}$ can be identified with an adjacency tensor~$h^{(\ell)}$ in the usual way:
For example, $h^{(3)}_{ijk} = 1$ if $\{i,j,k\}\in \mathcal{H}$.

A \emph{simplicial complex} is a hypergraph that satisfies the additional closure relation: 
If $H\in\mathcal{H}$ then for any $H'\subset H$ we have $H'\in\mathcal{H}$.
In this case, the hyperedge is a \emph{simplex}.
As an example, consider for a given graph with adjacency matrix~$(a_{ij})$ the simplicial complex one obtains by adding a simplex for each clique.
For a three-simplex~$\{i,j,k\}$ we then have $h^{(3)}_{ijk} = a_{ij}a_{jk}a_{ki}$.

\newcommand{\n}{\mathbf{n}}
\newcommand{\m}{\mathbf{m}}

% Here we consider reduced models of SIS contagion dynamics for $M$~populations on hypergraphs and simplicial complexes.
% The state of each population---typically describing the probability~$y_m\in[0,1]$ of a typical node in population~$m$ to be infected---is given by $y=(y_1, \dotsc, y_M)$.
% Write a multiindex $\n=(n_1, \dotsc, n_M)\in \N^M$ of order $\abs{\n} = n_1+\dotsb +n_m$ and set $y^\n = y_1^{n_1}\dotsb y_M^{n_M}$.
% Specifically, the SIS dynamics are given by ordinary differential equations
% \begin{equation}
% \label{eq:ModelGen}
%     \dot y_m = \sum_{\abs{\mathbf{n}}\leq D}C_\mathbf{n} y^\n -\gamma y_m
% \end{equation}
% %\IZK{$D$ is not defined, I think.}
% for $m=1, \dotsc, M$, a parameter $\gamma>0$, and a maximal polynomial degree $D<\infty$ that depends on the maximal interaction order.
% The coefficients~$C_\mathbf{n}$ to the properties of the underlying hypergraph.
% Note that each of the models~\eqref{eq:ex_iacopo}, \eqref{eq:ex_degree}, and~\eqref{eq:ex_kiss_insights} are special cases of~\eqref{eq:ModelGen}.
% In the following, we discuss in detail how models of the type~\eqref{eq:ModelGen} can be obtained from microscopic SIS dynamics and how the coefficients~$C_\mathbf{n}$ relate to the properties of the underlying hypergraph.

\subsection{SIS Dynamics of Individuals}

\newcommand{\Hc}{\mathcal{H}}

We consider SIS dynamics on hypergraphs and simplicial complexes, where the state $x_i(t)\in[0,1]$ of node~$i$ at time~$t$ corresponds to the probability of the node being infected. 
In classical graph-based SIS models, the infection pressure that node~$j$ exerts on node~$i$ is proportional to $a_{ij}(1-x_i)x_j$; the effect of all nodes in the neighborhood of~$i$ is the sum of the individual contributions. 
{{Note that this is in fact equivalent to closing the exact individual-based model at the level of pairs, where the joint probability of a node and its neighbours is approximated by the product of node-level probabilities.}}
For SIS dynamics on a hypergraph $\Hc$, write $\Hc_i = \{H\in\Hc\mid i\in H\}$ for the set of hyperedges that contain the node~$i$.
Generalizing the dynamics on graphs to include multi-body interactions, the infection pressure onto node~$i$ via a hyperedge $H\in\Hc_i$ is proportional to $(1-x_i) \tau_{H} \prod_{j\in H, j\ne i} x_j$, where $\tau_H$~is the infection rate corresponding to the hyperedge.
% \CB{Including itself?}
% \IZK{We usually do not consider self-loops.}
Hence, the individual-based mean-field model for SIS epidemic propagation on a hypergraph takes the form
\begin{equation}
\dot x_i = (1-x_i) \sum_{H\in {\cal H}_i} \tau_{H}  \prod_{j\in H,j\ne i} x_j - \gamma x_i, \label{individual based hypergraph}
\end{equation}
for $i=1,2,\dotsc, N$, where~$\gamma$ denotes the recovery rate common to all nodes.

%\IZK{Chris, I believe you said that you will also show maybe one eq for simplicial complexes. Using the $a_{ij}$ simplicial complex induced by graph thingy!}

%\CB{Maybe condense the following remark}

\begin{remark}
    First, while we consider polynomial interactions here, for a hyperedge~$H$ of order~$\ell$, one could also consider more general interaction functions $f^{(\ell)}$.
    Specifically, the interaction pressure through~$H\in\Hc_i$ onto~$i$ would be proportional to
    $(1-x_i) \tau_{H} f^{(\ell)}(x_H)$, where~$x_H$ is the vector of states with indices in~$H$.

    Second, while we here restrict ourselves to fairly simple interaction pressures, note that the infection pressure can be expressed in more sophisticated ways. 
As an example we mention \cite{de_arruda_social_2020}, where the differential equation for an individual is given in the form
\begin{equation*}
    \frac{dx_i}{dt}=-\delta x_i+\lambda(1-x_i)\sum_{j:v_i\in e_j}\sum_{k=\Theta_j}^{|e_j|-1}\lambda^*(|e_j|)\mathbb{P}_{e_j}^{v_i}(K=k) .
\end{equation*}
In this formulation the infection pressure is still in the stochastic (non mean-field) form enabling the dependence on the number of infected nodes in a hyperedge. 
%The mean-field approximation of this term can be found in \cite{de_arruda_social_2020}.
\end{remark}

In the following, we will focus on dynamics on hypergraphs with hyperedges up to order $\ell=4$.
Thus, the hyperedges consist of (a)~edges~$E\subset\Hc$ of order two corresponding to two-body interactions, (b)~triplets, or simply `triangles',~$T\subset\Hc$ of order three for the three-body interactions, and (c)~quadruplets, or simply `squares',~$S\subset\Hc$ of order four representing four-body interactions.
The sets~$E,T,S$ can be identified via an incidence matrix where the entry $(i,H)$ is one if node~$i$ is contained in the hyperedge~$H$.
For example, $E$ can be identified with a matrix of size $N\times |E|$ where~$|E|$ denotes the number of edges.
If the infection rate only depends on the edge order, that is, $\tau_\ell$~is the rate for any edge of order~$\ell$, the dynamical equations of $N$~individual nodes in SIS dynamics on a hypergraph can be written as 
\begin{equation}
    \begin{split}
\dot{x}_i &= \tau_2 (1-x_i) \sum_{j=1}^{\abs{E}} E_{ij} \sum_{k\neq i} x_k E_{kj}  + \tau_3 (1-x_i)\sum_{j=1}^{\abs{T}} T_{ij} \sum_{k,l\neq i, k<l} x_kx_l T_{kj}T_{lj}\\
 &\qquad+ \tau_4 (1-x_i)\sum_{j=1}^{\abs{S}} S_{ij} \sum_{k,l,m\neq i, k<l<m} x_kx_lx_m S_{kj}S_{lj}S_{mj} -\gamma x_i.
    \end{split}\label{individual based _hyp_graph}
\end{equation}
for $i=1,\ldots N$. 
We give a concrete example to illustrate specific instances of these general model equations.

\begin{example} \label{exam:hyp_graph}
A fully connected graph with 4 nodes, but with only two triangles: $N=4$,
\begin{itemize}
\item $E=\{ (1,2), (1,3), (1,4), (2,3), (2,4), (3,4)\}$, $|E|=6$,
\item $T=\{ (1,2,3), (1,3,4)\}$, $|T|=2$,
\item $S=\{(1,2,3,4)\}$, $|S|=1$.
\end{itemize}
The first equation of the individual based  model for this hypergraph takes the form
\[
\dot{x}_1 = \tau_2 (1-x_1) (x_2+x_3+x_4) + \tau_3 (1-x_1) (x_2x_3 + x_2x_4) + \tau_4 (1-x_1) x_2x_3x_4 - \gamma x_1 .
\]
The other equations follow in a similar way and we do not give them explicitly.
\end{example}

%\IZK{Example with a simplicial complex induced by a graph. IN fact example two is a simplicial complex thiking.}
%\CB{As an additional example?}

For simplicial complexes induced by an underlying graph with adjacency matrix~$A=(a_{ij})_{i,j=1,\dotsc, N}$, the coefficients of the adjacency tensor can be expressed in terms of entries of $A$.
Assuming that three- and four-body interactions are considered (i.e., all triangles and fully connected 4-cliques in a classical network formulation become a three- and four-body interaction, respectively), the individual based mean-field model takes the form 
\begin{equation}
    \begin{split}        
    \dot{x}_i&= \tau_2  (1-x_i)\sum_{j=1}^N a_{ij}x_j + \tau_3 (1-x_i)\sum_{j=1}^N \sum_{k=j+1}^N a_{ij}a_{ik}a_{jk}x_jx_k \\
    &\qquad+ \tau_4 (1-x_i) \sum_{j=1}^N \sum_{k=j+1}^N \sum_{l=k+1}^N a_{ij}a_{ik}a_{il}a_{jk}a_{jl}a_{kl}x_jx_kx_l -\gamma x_i,
    \end{split}    
    \label{individual_based_simpl_compl}
\end{equation} 
for $i=1,2,\ldots , N$. It is worth noting that in this particular example, we chose to neglect interactions among five or more nodes even if the corresponding fully connected cliques are present.

\begin{example}
As an example for a simplicial complex induced by a graph, consider the regular ring lattice with~$N$ nodes, where each node is connected to their~$n$ nearest neighbours. 
The node-level equations are given by
\begin{equation}
\begin{split}
    \frac{dx_i}{dt}&=\tau_2 (1-x_i)\sum_{j=1}^{n} x_{i+j}+ \tau_2  (1-x_i)\sum_{j=1}^{n} x_{i-j}\\
    &\qquad+\tau_3 (1-x_i)\sum_{j=i+1}^{i+n-1}\sum_{k=j+1}^{i+n} x_jx_k + \tau_3 (1-x_i)\sum_{j=i-1}^{i-n+1}\sum_{k=j-1}^{i-n} x_jx_k\\
    &\qquad+ \tau_4 (1-x_i)\sum_{j=i+1}^{i+n-2}\sum_{k=j+1}^{i+n-1}\sum_{l=k+1}^{i+n} x_jx_kx_l + \tau_4 (1-x_i)\sum_{j=i-1}^{i-n+2}\sum_{k=j-1}^{i-n+1}\sum_{l=k+1}^{i-n} x_jx_kx_l-\gamma x_i,
    \label{eq:individual based _reg_ring_latt_up_to_sq}
\end{split}
\end{equation}
where running indices are to be considered modulo~$N$. 
In contrast to Example \ref{exam:hyp_graph}, we note that on the ring-lattice, the 1-simplices, 2-simplices of every 3-simplex partakes in the dynamics. 
However, in that example, not all 2-simplices are present despite the existence of a hyperedge encompassing all nodes.
\end{example}

% \textbf{Simplicial complexes} An important special case of a hypergraph is when we are given a traditional network by its adjacency matrix $A$ and the edges, triangles and squares are simply given by the graph itself. This can be formulated as follows.

% \begin{itemize}
% \item Edges: a pair of nodes $(i,j)$ is in the edge set $E$, if $a_{ij}=1$.
% \item Triangles: a triple of nodes $(i,j,k)$ is in the triangle set $T$, if $a_{ij}a_{ik}a_{jk}=1$.
% \item Squares: a quadruple of nodes $(i,j,k,l)$ is in the set $S$, if $a_{ij}a_{ik}a_{il}a_{jk}a_{jl}a_{kl}=1$.
% \end{itemize}

% \textbf{The individual based  model} on a simplicial complex can be written in the following form for all $i=1,\ldots N$.
% \begin{eqnarray}
% \dot{x}_i &=& \tau (1-x_i) \sum_{j} a_{ij} x_j  + \beta (1-x_i)\sum_{j,k} a_{ij}a_{ik}a_{jk} x_jx_k \\
%  &+& \delta (1-x_i)\sum_{j,k,l} a_{ij}a_{ik}a_{il}a_{jk}a_{jl}a_{kl} x_jx_kx_l -\gamma x_i . \label{individual based _simp_comp}
% \end{eqnarray}

\subsection{Low-dimensional reduction and examples} 
\label{subsec:lowdim}

Since the individual-based mean-field equations for a large number of nodes~$N$ are challenging to analyse, further dimension reduction must be employed. 
% mean-field approximations provide a suitable way to make statements about the evolution on a population level.
% To obtain low-dimensional equations, the starting point are typically equations at node level where the probability of a node being infected depends on joint probability of the node in question and the state of its immediate neighbours. 
% To obtain a one dimensional mean-field model one must assume that joint probabilities over pairs or simplicity complexes can be written as products of the states of individual nodes, a strong independence assumption. 
This can be carried out by assuming homogeneity: Each node has the same number of neighbours and that it is part of the same number of two-body, three-body, and higher-order multi-body interactions; the actual number of interactions can be and it is usually different for different interaction orders.
This allows us to reduce the $N$-dimensional system to one single differential equation and the same argument was
% interactions allows to reduce a system of equations to one single one~\cite{kiss2017mathematics}.
 used to derive equations including~\eqref{eq:ex_iacopo}, \eqref{eq:ex_degree} and \eqref{eq:ex_kiss_insights}. 
In the following, we discuss such reductions to low-dimensional systems
%to systems of the form~\eqref{eq:ModelGen} 
in more detail. 

\begin{figure}
    \centering
    \includegraphics[width=\linewidth]{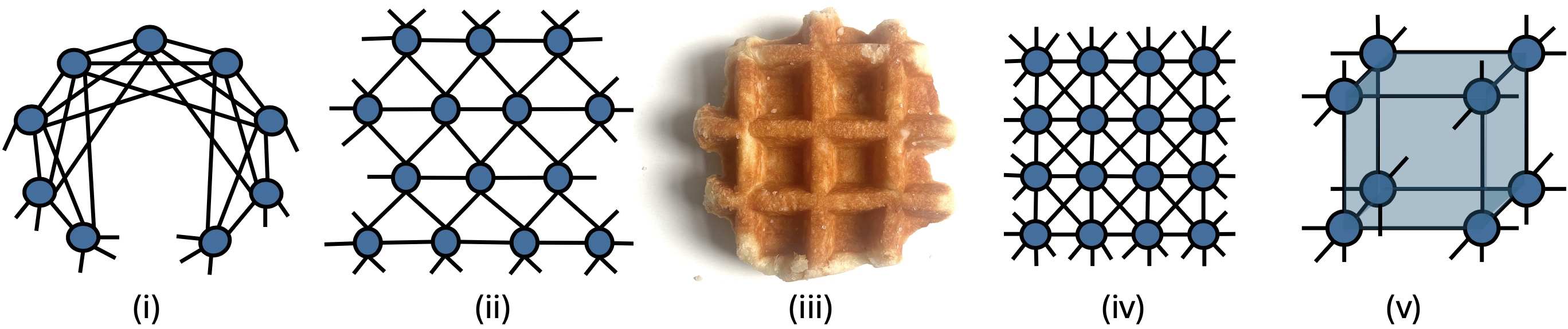}
    \caption{Sketches of homogenous networks with topologically equivalent nodes. 
   From left to right we have: (i)~ring lattice, (ii)~triangular lattice, (iii)~square lattice, (iv)~square lattice with diagonals, and (iv)~cube lattice with hyperedges along the faces. 
    Periodic boundaries are assumed for each. 
    Despite the homogeneity, the networks can be differentiated by the number of two-, three-, four-body interactions that affect each node.}
    \label{fig:many_networks}
\end{figure}

Homogeneity allows us to go from individual-based mean-field models to a single differential equation; we illustrate the main steps here. 
Let us take the ring lattice with each node connecting to~$n$ nearest neighbours on the right and left, see top left panel in Figure~\ref{fig:many_networks}. 
Further we consider that each triangle and fully connected square that arises are in fact simplicial complexes of orders~2 and~3, respectively. The individual-based mean-field model for such an underlying contact network is given in \eqref{eq:individual based _reg_ring_latt_up_to_sq}.
It now follows immediately that each node is contained in $2n$~one-simplices,  $n(n-1)+1$~two-simplices and $n(n-1)(n-2)/3+2$~three-simplices. 
Moreover, if the initial conditions in all nodes are the same, all variables of the system evolve identically to one another.
Thus, writing $y(t)=x_i(t)$ for each coordinate and taking into account that all nodes belong to the same number of two-body, three-body and four-body interactions, 
the individual level model~\eqref{eq:individual based _reg_ring_latt_up_to_sq} can be reduced to one single equation, namely
\begin{align}
\label{eq:reg_ring_latt_up_to_sq}
    \frac{dy}{dt}=\tau \underbrace{(2n)}_{\text{\# of edges}}(1-y)y+\beta\underbrace{(n(n-1)+1)}_{\text{\# of triangles}}(1-y)y^2+\delta\underbrace{\left(\frac{n(n-1)(n-2)}{3}+2\right)}_{\text{\# of squares}}(1-y)y^3 -\gamma y.
\end{align}
This general approach can be extended to other networks as long as each node belongs to the same number of simplices of different order. 
By introducing~$d$,~$r$ and~$q$ as the number of two-body, three-body and four-body interactions, respectively, the equation above can be genralised to
% Then all~$x_i$ are equal to each other and we work with their common value $y(t)=x_i(t)$. Going up to the level of 4-body interactions the differential equation for~$y$ takes the form
\begin{equation}
\dot y(t) = \tau d (1-y) y + \beta r (1-y)y^2 + \delta q (1-y) y^3 -\gamma y.
\label{eq:general_d_r_q_intuitive}
\end{equation}
% where $d$, $r$ and $q$ are the number of 2-body, 3-body and 4-body interactions, respectively.
% which corresponds to $C_2=\tau d$, $C_3=\beta r$ and $C_4=\delta q$. 
This opens up the possibility to consider many different networks and it increases the range of values the rate parameters can take. Alternatively, if graphability is not an issue and stochastic simulations on an explicit network are not desired, then the system can be studied with arbitrary rates.
To illustrate the flexibility of this approach, below we list further choices for~$d$,~$r$ and~$q$  for a range of networks, some illustrated in Figure~\ref{fig:many_networks}. In no particular order options include,
\begin{enumerate}[(i)]
\item The ring lattice with links to the $n$ nearest neighbours on left and right: $d=2n$, $r=n(n-1)+1$, $q=\frac{n(n-1)(n-2)}{3}+2$,
\item Triangle lattice on a torus: $d=6$, $r=6$, $q=0$,
\item Square lattice on a torus: $d=4$, $r=0$, $q=0$,
\item Square lattice with diagonals on a torus: $d=8$, $r=12$, $q=4$,
\item Three dimensional cubic lattice with squares as hyperedges: $d=6$, $r=0$, $q=12$,
%\item Fully connected network of size $N$:
$d=(N-1)$, $r={N-1\choose 2}$, $q={N-1 \choose 3}$.
% \item Mean-field limit of the complex contagion on the fully connected network up to 4-body interactions: ~\cite{kiss_insights_2023}
\end{enumerate}
% \IZK{THERE IS AN ISSUE HERE! Repetition etc}
% \CB{This is the same text that was item (vi); see commented out item in list. Can leave it out if you think it is repetitive.}
Note that a fully connected network of size~$N$ has parameters $d=(N-1)$, $r={N-1\choose 2}$, $q={N-1 \choose 3}$ but corresponds to Model~C. The choices above and those in models given by equations~\eqref{eq:ex_iacopo}, \eqref{eq:ex_degree} and~\eqref{eq:ex_kiss_insights} are summarised in Table~\ref{tab:examples_of_poly_coeff}. 

From a general dynamical systems point of view, the evolution equations can be seen as polynomial equations.
Indeed, reordering terms, the population level dynamics can be written as 
\begin{align}
    \frac{dy}{dt}&=(1-y)\sum_{q=1}^{Q}C_{q+1}y^{q} -\gamma y \nonumber\\
    &=(C_2-\gamma)y + \sum_{q=2}^{Q}(C_{q+1}-C_{q})y^{q}  -C_{Q+1}y^{Q+1},    \label{eq:general_one_pop_ODE_with_Ci}
\end{align}
where the coefficients~$C_q$ are typically a product between the transmission rate within a simplex/hyperedge with $q$-body interactions and the number of such simplices a typical node belongs to. 
However, when formulating models we find that terms are better understood if the alternative formulation in equation~\eqref{eq:general_d_r_q_intuitive} is used. Both formulations are used throughout the paper.

\begin{table}%[H]
\begin{center}
\scriptsize{
\begin{tabular}{ |c|c|c|c|c|} 
\hline
        & $\mathbf{y^1}$        & $\mathbf{y^2}$         & $\mathbf{y^3}$  &                 $\mathbf{y^4}$\\
        \hline
{\textbf{General}} & $\mathbf{C_2}$-\boldmath$\gamma$  &$\mathbf{C_3-C_2}$      & $\mathbf{C_4-C_3}$ or $\mathbf{-C_3}$     & $\mathbf{-C_4}$ \\
\hline
{{A}} &  $\beta_{1}\langle k_{1}\rangle-\mu$& $\beta_{2}\langle k_{2}\rangle-\beta_{1}\langle k_{1}\rangle$ & $\beta_{3}\langle k_{3}\rangle-\beta_{2}\langle k_{2}\rangle$ & $-\beta_{3}\langle k_{3}\rangle$\\
\hline
{{B}} & $\beta_{\mathrm{eff}}-1$ & $\gamma \beta_{\mathrm{eff}}-\lambda \beta_{\mathrm{eff}}$ & $-\gamma \beta_{\mathrm{eff}}$& \\
\hline
{{C}} & $s_1 -\gamma$ & $ s_2/2- s_1$ & $s_3/6-s_2/2$ or $-s_2/2$& $-s_3/6$\\
\hline
{(i)} & $\tau \times 2n -\gamma$ & $\beta (n(n-1)+1)-\tau \times 2n$ & $\delta \left(\frac{n(n-1)(n-2)}{3}+2\right)- \beta (n(n-1)+1)$ & $-\delta \left(\frac{n(n-1)(n-2)}{3}+2\right)$ \\
\hline
{(ii)} & $\tau \times 6 -\gamma$ & $\beta \times 6 - \tau \times 6$ & $-\beta \times 6$ & \\
\hline
{(iii)} & $\tau \times 6 -\gamma$ & $ - \tau \times 6$ &  & \\
\hline
{(iv)} & $\tau \times 8 -\gamma$ & $\beta \times 12 - \tau \times 8$ & $\delta \times 4-\beta \times 12$ & $-4 \times \delta$ \\
\hline
{(v)} & $\tau \times 6 -\gamma$ & $ - \tau \times 6$ & $+\delta \times 12$ & $-\delta \times 12$\\
\hline
\end{tabular}
}
\end{center}
\caption{
Interaction parameters---including the presence of higher-order interactions---determine the order and coefficients of the polynomial contagion dynamics~\eqref{eq:general_one_pop_ODE_with_Ci}
for homogeneous networks.
Rows~A, B, and~C correspond to models~\eqref{eq:ex_iacopo}, \eqref{eq:ex_degree}, and~\eqref{eq:ex_kiss_insights}, respectively.
Rows~(i--v) give the parameters for a
(i)~ring lattice, (ii)~triangular lattice, (iii)~square lattice, (iv)~square lattice with diagonals, and (v)~cube lattice shown in Figure~\ref{fig:many_networks} and discussed in detail in the main text.
%Note that~C differs in that unlike simplicial complexes, hyperedges do not imply the existence of all lower-order interactions but the resulting one-dimensional mean-field model still has a polynomial right hand side which is distinct from the general equation~\eqref{eq:general_one_pop_ODE_with_Ci}. 
}
\label{tab:examples_of_poly_coeff}
\end{table}

\subsection{Multi-population models}

If one assumes homogeneity in a particular population of individuals, rather than of all individuals, one obtains an equation for each population.
Consider~$M$ populations, where~$y_m\in[0,1]$ describes the probability of of a typical node in population~$m$ to be infected and $y=(y_1, \dotsc, y_M)$ is the joint state of all populations.
Write a multiindex $\n=(n_1, \dotsc, n_M)\in \N^M$ of order $\abs{\n} = n_1+\dotsb +n_m$ and set $y^\n = y_1^{n_1}\dotsb y_M^{n_M}$.
Then the multipopulation SIS dynamics in a very general form are given by ordinary differential equations with polynomial vector field
\begin{equation}
\label{eq:ModelGen}
    \dot y_m = \sum_{\abs{\mathbf{n}}\leq L}C_\mathbf{n} y^\n -\gamma y_m
\end{equation}
for $m=1, \dotsc, M$ and a parameter $\gamma>0$.
The maximal polynomial degree $L<\infty$ that depends on the maximal order of the multibody interactions.
Moreover, the coefficients~$C_\mathbf{n}$ are directly relates to the properties of the underlying hypergraph as for $M=1$ population above.

\begin{figure}
    \centering
\includegraphics[scale=0.5,trim={8cm 3.5cm 8cm 3.5cm},clip]{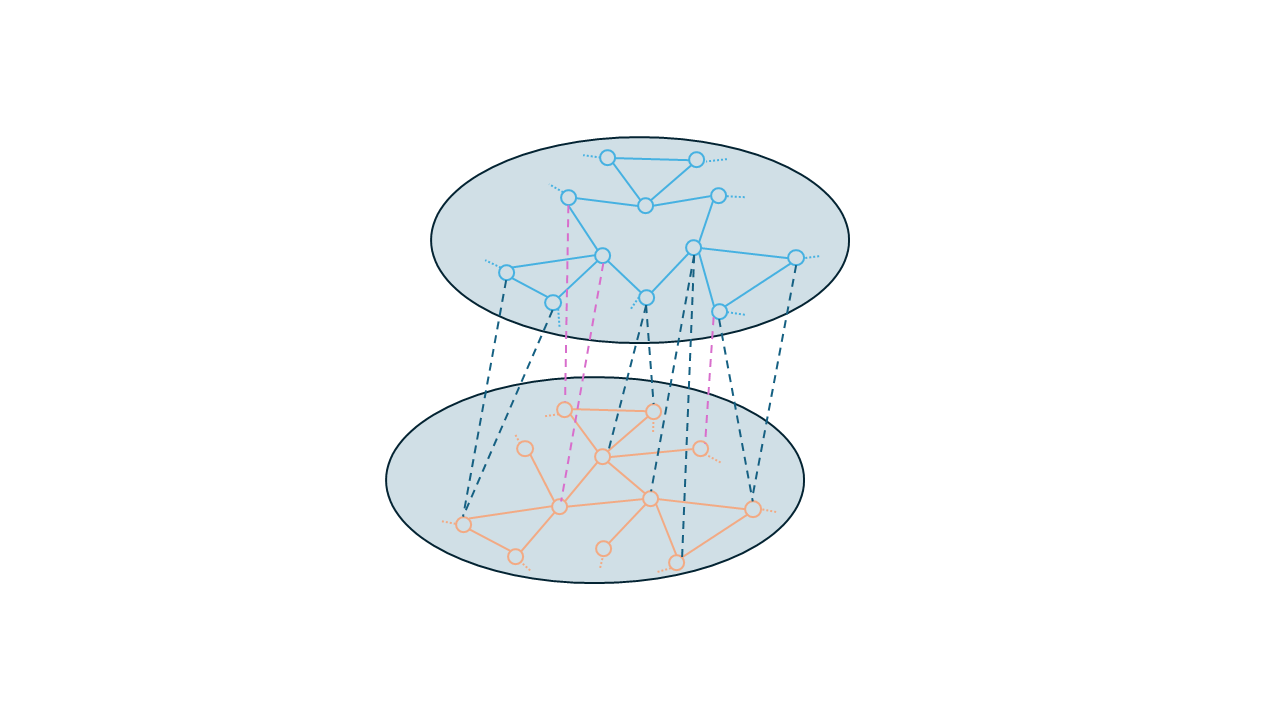}
    \caption{Sketch of a network consisting of two populations of topologically equivalent nodes.
    Each node in a given population is incident to the same number of edges (two in the upper population, one in the lower) and three-body interactions (one in the top and two in the lower) with further pairwise (dashed cyan) and 
    three-body interactions (dashed blue) between the two populations. Note that links within a three-body interaction (three-clique) do not count towards edges/pairwise interactions. 
    Dotted stubs correspond to further interactions within each population.
    The interactions between the two populations can be modulated by considering only two-body interactions, only three-body interactions, or both types of interactions. }
    \label{fig:two_pop_illustration}
\end{figure}

A network consisting of $M=2$ populations is depicted in Figure~\ref{fig:two_pop_illustration}.
We analyse bifurcations in such a two-population model explicitly in Section~\ref{sec:TwoPop}.

%Note that each of the models~\eqref{eq:ex_iacopo}, \eqref{eq:ex_degree}, and~\eqref{eq:ex_kiss_insights} are special cases of~\eqref{eq:ModelGen}.
%In the following, we discuss in detail how models of the type~\eqref{eq:ModelGen} can be obtained from microscopic SIS dynamics and how the coefficients~$C_\mathbf{n}$ relate to the properties of the underlying hypergraph.

%%%%%%%%%%%%%%%%%%%%%%%%%%%%%%%%%%%%%%%%%%%%%%%%%%%%%%%%%
\section{Complete bifurcation regimes for single population models}
%%%%%%%%%%%%%%%%%%%%%%%%%%%%%%%%%%%%%%%%%%%%%%%%%%%%%%%%%
\label{sec:results}
We now focus on the simplest class of models whose dynamics are given by a single differential equation, that is,~\eqref{eq:ModelGen} for $M=1$.
This includes networks where all nodes are topologically equivalent.
We identify all possible qualitative dynamics of these models and determine the bifurcations that lead to bistability, hysteresis, multistability, and how they depend on the network parameters.

First, recall the general classification of the transcritical bifurcation and its criticality.
Specifically, consider a one-dimensional equation 
\begin{equation}
\label{eq:general_1d_ODE}
    \dot y =f(y,\lambda),
\end{equation}
that depends on a parameter~$\lambda$.
Assume that $y=0$ is an equilibrium for all~$\lambda$, that is, $f(0,\lambda)=0$.
If $\lambda^{*}$ is a solution of 
    $\partial_{y}f(0,\lambda)=0$
and
$ \partial_{\lambda y}f(0,\lambda^{*})\ne 0$, $  \partial_{y y}f(0,\lambda^{*})\ne 0,$
then the system undergoes a transcritical bifurcation at~$\lambda^*$.
Depending on the sign of the higher derivatives, the transcritical bifurcation is one of the four types depicted in Figure~\ref{fig:trans_crit_class}.

\begin{figure}[h!]
     \centering   
     \includegraphics[width=0.24\linewidth]{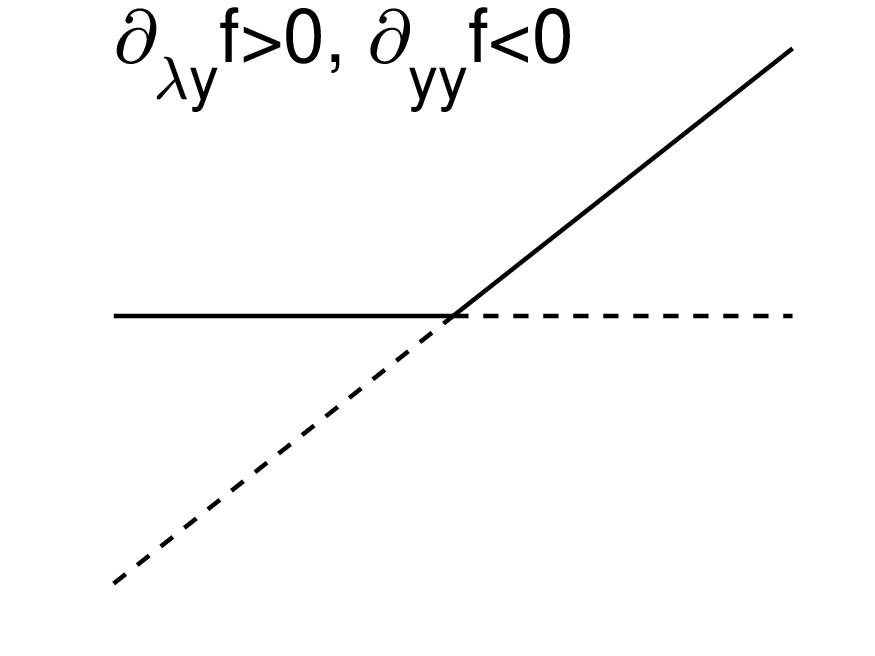}
     \includegraphics[width=0.24\linewidth]{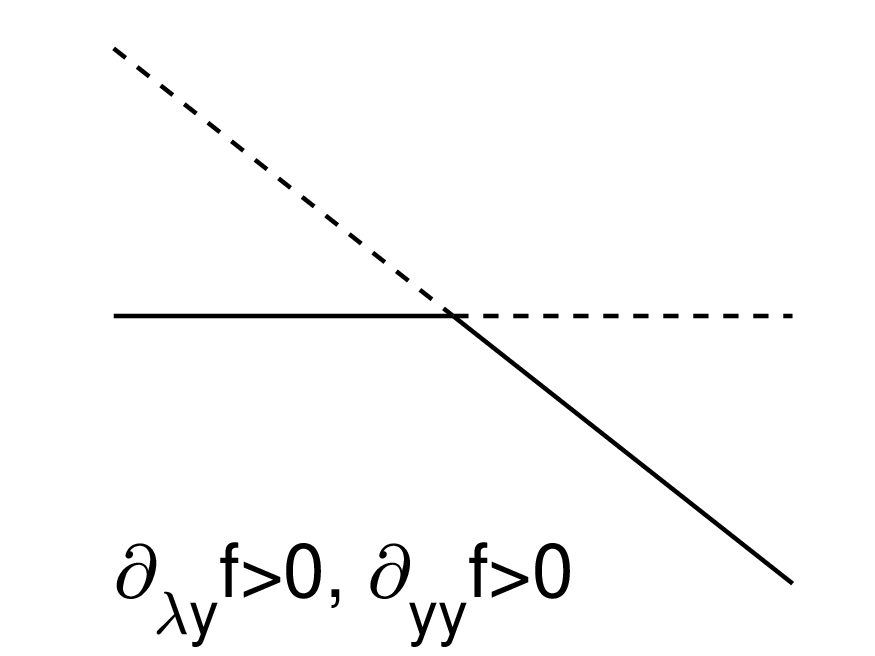}
     \includegraphics[width=0.24\linewidth]{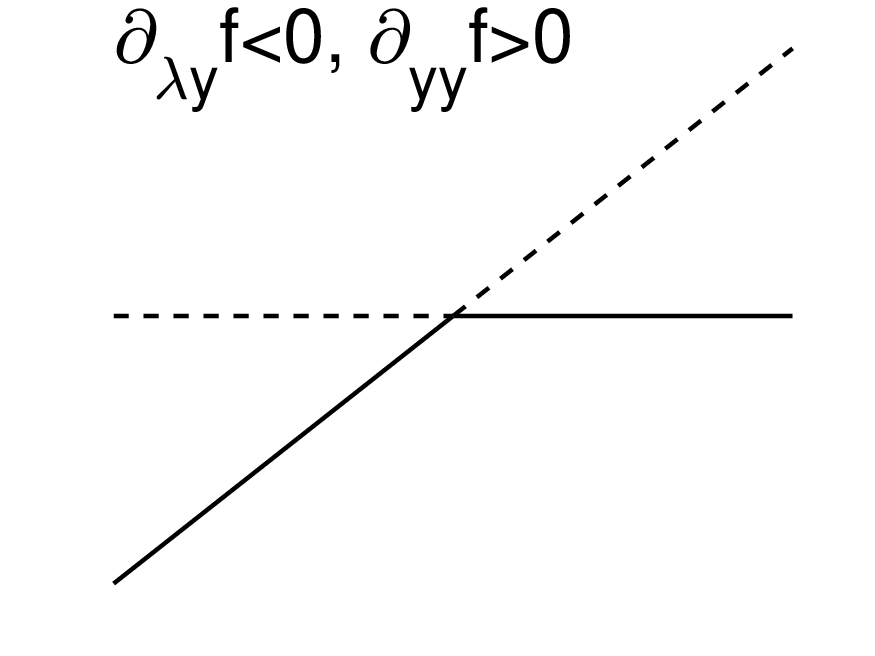}
     \includegraphics[width=0.24\linewidth]{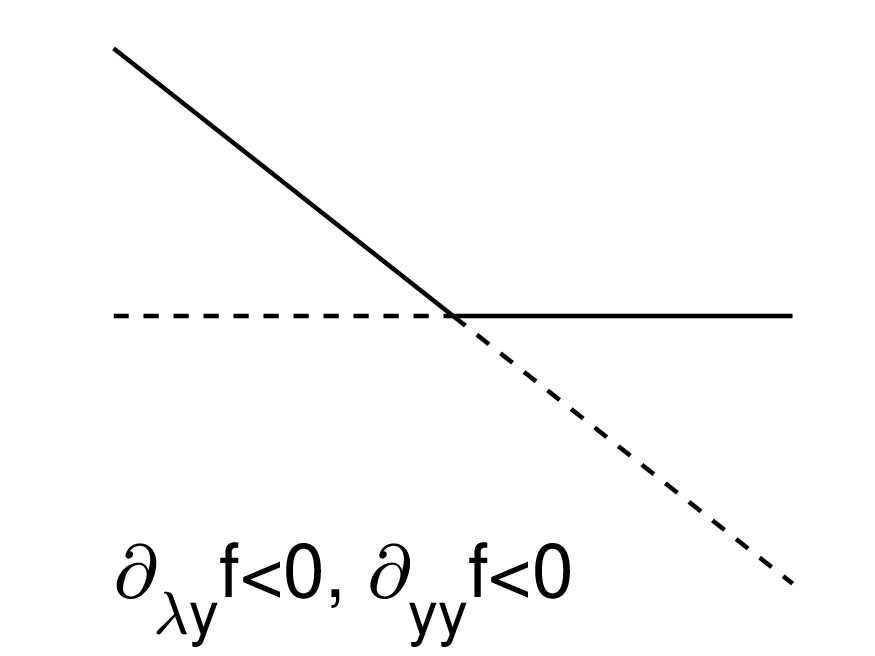}
     \caption{Illustration of possible  transcritical bifurcation classes based on the partial derivatives of the general vector field in~\eqref{eq:general_1d_ODE} at the bifurcation point. The horizontal axis represents the bifurcation parameter~$\lambda$ and the vertical axis denotes the state variable $y$. 
     Stable branches are depicted by a solid line, while unstable branches correspond to dashed lines. 
     % \textcolor{red}{smaller and in one row?}\,\,\, {\color{green} took me bloody 20 minutes or so, thank you!}
     }
     \label{fig:trans_crit_class}
\end{figure}

%%%
\subsection{Complex contagion up to three-body interactions}
\label{sec:OnePopTriangles}

For a one-dimensional mean-field equation, the derivatives of the vector field now determine the possible bifurcations of the disease free state. We first consider interactions given by two-body and three-body interactions.

\begin{prop}
    \label{lemma:2_simplices}
    For a disease transmission model of the form
    \begin{equation}
    \label{eq:CubicModel}
        \dot y=(C_2-\gamma)y+(C_3-C_2)y^2-C_3y^3
    \end{equation}
    with $C_2\neq C_3$
    the disease-free steady state undergoes a transcritical bifurcation when $C_2=\gamma$. Furthermore, this can be of two types: (i)~if $C_3<\gamma$ then a forward bifurcation results, and (ii)~if $C_3>\gamma$ then the transcritical bifurcation is of backward type and the systems exhibits bistability as shown in Figure~\ref{fig:pos_bif_diag_upto_tri}.
\end{prop}

\begin{proof}
%For the proof it is important to note that bifurcations parameters are incorporated in the coefficients~$C_2$ and~$C_3$.
%The classification of the transcritical bifurcation in terms of the parameters now yields the type of bifurcation locally at the bifurcation point.

According to the general theorem about the transcritical bifurcation and using the sign conditions in Figure \ref{fig:trans_crit_class}, we have forward transcritial bifurcation at $C_2=\gamma$, when $C_3<\gamma$. Similarly, the bifurcation is backward when $C_3>\gamma$.

The global nature of the curve can be obtained in this simple case by expressing the parameter~$C_2$ in terms of~$y$ as
\[
C_2= \frac{\gamma}{1-y} -C_3 y.
\]
Differentiating~$C_2$ with respect to~$y$ yields that the bifurcation curve has a fold point (saddle-node bifurcation) with $C_2'(y)=0$ when $\gamma/C_3 = (1-y)^2$. 
This fold point lies in the relevant domain ($0<y<1)$ when $C_3>\gamma$, i.e., when the transcritical bifurcation is of backward type. This completes the proof of the statement.
\end{proof}

\begin{remark}
    If $C_2=C_3$, then the quadratic term vanishes, leading to a $y\mapsto-y$ symmetry. 
    This means that the trivial state loses stability in a pitchfork bifurcation.
    Furthermore, this is an example of how three-body interactions can change the type of pitchfork bifurcation; cf.~\cite{kuehn2021universal}.
\end{remark}

\begin{figure}[h!]
     \centering
     \includegraphics[scale=0.2]{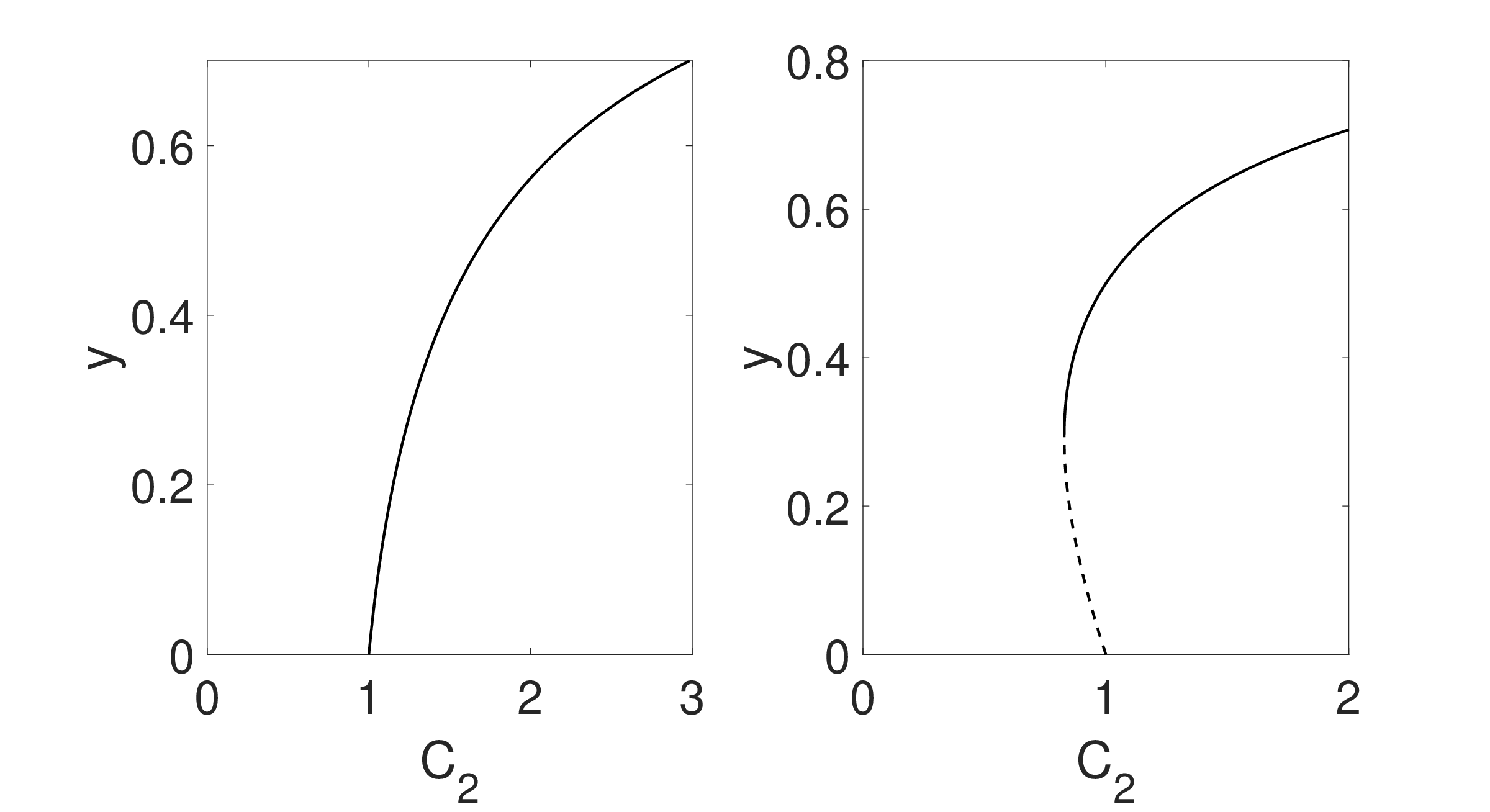}
     \caption{Possible outcomes for a system such as given in equation~\eqref{eq:CubicModel}. Parameters are $C_3=0.5$ (left) and $C_3=2$ (right) with  $\gamma=1$ in both cases.}
     \label{fig:pos_bif_diag_upto_tri}
\end{figure}

Proposition~\ref{lemma:2_simplices} now allows us to classify the dynamics in terms of the network parameters.
For a ring lattice as given in equation~\eqref{eq:reg_ring_latt_up_to_sq} with infection through triangles or 2-simplices. 
Consider~$\tau$ as the bifurcation parameter. 
% \CB{$r_1$?}
% \IZK{$r_1$ changed to $\tau$, nice one for noticing.}
Rearranging the equation as per Proposition~\ref{lemma:2_simplices} above leads to
\begin{align}
%\label{eq:reg_ring_latt_up_to_sq}
    \dot y =(\tau \times 2n-\gamma)y + \left(\beta \times (n(n-1)+1)-\tau \times 2n \right)y^2-\beta 
 \times (n(n+1)+1)y^3
\end{align}
According to Proposition~\ref{lemma:2_simplices}, the bifurcation occurs when $2\tau n=\gamma$ with bistability if 
\begin{equation}
\beta (n(n-1)+1)-\gamma>0,
\end{equation}
which is equivalent to 
\begin{equation}
n(n-1)>(\gamma-\beta)/\beta.
\end{equation}
For fixed transmission rates through links and 2-simplices as well as fixed recovery, the condition above gives a constraint on the density of links, and implicitly 2-simplices, needed for bistability. The immediate interpretation of this result is that the same disease unfolding on communities with different contact patterns can lead to different outcomes.

Our result can also be used to re-derive the bifurcation conditions for the models cited in Section~\ref{subsec:lowdim}. For example, for the model up to 2-simplices in~\cite{kiss_insights_2023}, our Proposition predicts that the bifurcation occurs at $C_2=\gamma$, that is $\lambda=\gamma$. For bistability our Proposition requires $C_3>\gamma$ which translates to $\mu/2<\gamma$ which is exactly as reported in~\cite{kiss_insights_2023}.

%%%
\subsection{Complex contagion up to four-body interactions}
\label{sec:HomogeneousMultistab}
% \IZK{I would no longer write triangles and squares, also hyperedges hardly appear anywher after model definition. The reader may say that what we do is almost exclusively for simplicies. Can we delete trianlges and square in these parts?}
We now consider interactions given by two-body, three-body and four-body interactions and determine the bifurcations.
The different possibilities for the transcritical bifurcation characterize the possible contagion dynamics.

\begin{proposition}
\label{lemma:3_simplices}
    For a disease transmission model of the form
    \begin{equation}
     \label{eq:QuarticModel}
         \dot y =(C_2-\gamma)y+(C_3-C_2)y^2+(C_4-C_3)y^3-C_4y^4
    \end{equation}
    with $C_2\neq C_3$ the disease-free steady state undergoes a transcritical bifurcation when $C_2=\gamma$. 
    This can be of two types: (i) if $C_3<\gamma$ then a forward bifurcation results, and (ii) if $C_3>\gamma$ then the transcritical bifurcation is of backward type. 
    The systems exhibits three different global bifurcation curves:
    \begin{itemize}
        \item If $C_3>\gamma$, then the bifurcation curve has a fold point leading to bistability between the disease-free and the endemic steady state for certain values of the parameter $C_2$.
        \item If $C_3<\gamma$ and $C_4<\gamma$, then the bifurcation curve has no fold point, hence either the disease-free or the endemic steady state is stable for any value of the parameter $C_2$.
        \item If $C_3<\gamma$ and $C_4>\gamma$, then the bifurcation curve has two fold points, leading to two stable endemic branches. That is, bistability may occur between two endemic steady states for certain values of the parameter $C_2$.
    \end{itemize}
    The different cases are shown in Figure~\ref{fig:pos_bif_diag_upto_sq}. We note that two subfigures are shown for the case $C_3>\gamma$. The difference between them is in the value of~$C_4$, namely, there is an inflection point on the unstable branch when $C_4>\gamma$.
\end{proposition}

\begin{figure}%[h!]
     \centering
     \includegraphics[scale=0.4]{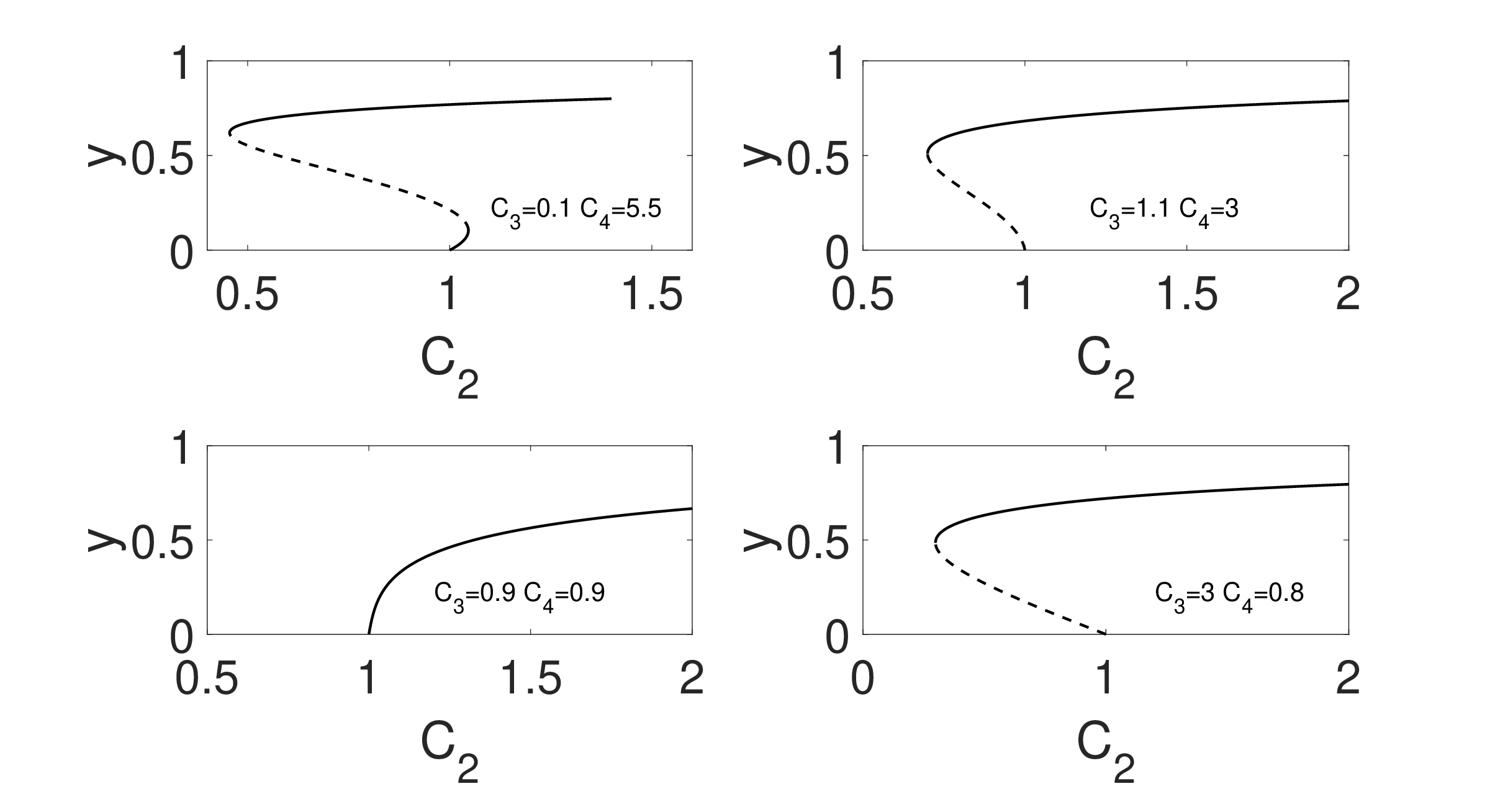}
     \caption{Possible bifurcation behaviour of the disease free steady state and emergent nontrivial equilibrium branches for~\eqref{eq:QuarticModel}. Parameter values are given in each panel and $\gamma=1$ for all cases.}
     \label{fig:pos_bif_diag_upto_sq}
\end{figure}

\begin{proof}
According to the general theorem about the transcritical bifurcation and using the sign conditions in Figure~\ref{fig:trans_crit_class}, we have forward transcritial bifurcation at $C_2=\gamma$, when $C_3<\gamma$. Similarly, the bifurcation is backward when $C_3>\gamma$. 

The global nature of the curve can be obtained in this simple case by expressing the parameter~$C_2$ in terms of~$y$ as
\[
C_2= \frac{\gamma}{1-y} -C_3 y - C_4 y^2.
\]
Differentiating $C_2$ with respect to $y$ yields that the bifurcation curve has a fold point, $C_2'(y)=0$, when $C_3 +2C_4 y = \gamma / (1-y)^2$. The shape of the bifurcation curve depends on the second derivative~$C_2''(y)$ as well. 
It has an inflection point when $\gamma /C_4 = (1-y)^3$. That is we have an inflection point when $C_4>\gamma$, leading to the bifurcation curves shown in Figure~\ref{fig:trans_crit_class} and completing the proof of the statement.
\end{proof}

We now use this result to determine the bifurcation behaviour for different networks.
Specifically, consider the one-dimensional equation~\eqref{eq:general_d_r_q_intuitive} which corresponds to equation~\eqref{eq:QuarticModel} with parameters $C_2=\tau d$, $C_3 = \beta r$, and  $C_4= \delta q$.
Fix the epidemic parameters such that
% \IZK{Should $\gamma$ be given?}
\[
\frac{1}{12} < \beta < \frac{1}{6} , \qquad \frac{1}{12} < \delta < \frac{1}{4}\,\,\, \text{and} \,\,\, \gamma=1.
\]

For the square lattice with diagonals on a torus we have $d=8$, $r=12$, $q=4$.
Thus, $C_3=12 \beta >1$ and $C_4=4 \delta <1$ and the bifurcation diagram is like that in the bottom right of Figure~\ref{fig:pos_bif_diag_upto_sq}. 
That is bistability occurs between the disease-free and the endemic steady state.

The triangle lattice on a torus is characterized by $d=6$, $r=6$, $q=0$. This implies $C_3=6 \beta <1$ and $C_4=0 <1$, and thus the bifurcation diagram is like that in the bottom left of Figure~\ref{fig:pos_bif_diag_upto_sq}.
That is, there is no bistability for this network. 

For the three dimensional cubic lattice with squares as hyperedges we have $d=6$, $r=0$, $q=12$. Thus, $C_3=0 <1$ and $C_4=12 \delta >1$. The bifurcation diagram is like that in the top left of Figure~\ref{fig:pos_bif_diag_upto_sq}. That is bistability may occur between two endemic steady states.

\section{Bifurcation regimes for two-population networks}
\label{sec:TwoPop}
%\IZK{Potentially can/should be moved to section 2.3. LEAVE IT AS IS!}
Rather than considering a homogeneous population of nodes, we now shift the focus to networks that consist of two distinct populations, each consisting of topologically equivalent nodes.
More precisely, consider a network with two populations of nodes---illustrated in Figure~\ref{fig:two_pop_illustration}.
Within population~1, each node has $d_1$~edges and $h_1$~triangles as well as $d_2$~edges to nodes in population~2.
Within population~2, each node has $d_4$~edges and $h_2$~triangles as well as $d_3$~edges to nodes in population~1.
For homogeneous initial conditions, this results in the evolution of the probability of infection~$y_j$ of a node in population~$j$ given by
\begin{subequations}  
\label{individual based _2groups}
\begin{align}
\dot y_1 &= \tau (1-y_1) (d_1y_1 + d_2 y_2) +\beta h_1 (1-y_1) y_1^2-\gamma y_1 ,\\
\dot y_2 &= \tau (1-y_2) (d_3y_1 + d_4 y_2) +\beta h_2 (1-y_2) y_2^2 -\gamma y_2.
\end{align}
\end{subequations}
This is a specific case of the general model equations~\eqref{eq:ModelGen} for $M=2$ populations.

In this section, we consider the possible bifurcation diagrams in the $(\tau, I)$-plane, where $I=N_1 y_1 + N_2 y_2$ with  $N_i$~denoting the number of nodes in population $i$.
% \CB{$N_i$ is the number of nodes per population?}
In contrast to one-population networks, where multistability requires the existence of 3-simplices (cf.~Section~\ref{sec:results}), we show that multistability is already possible with 2-simplices (triangles)---at the expense of a more complicated network structure.
We fist consider two symmetrically coupled populations and then analyse the dynamics if the symmetry is broken.

%%%
\subsection{Two symmetrically coupled populations}

Assume symmetric coupling, that is, $d_1 = d_4$, $d_2=d_3$, and $h_1=h_2$. 
Then the system has a $\Z_2$-symmetry that acts by permuting the indices of the populations.
The symmetry implies that the set where the states of both populations are equal, $\Delta := \{y = y_1 = y_2\}$, is dynamically invariant; cf.~\cite{Golubitsky2002}.
Thus, to analyse the dynamics of~\eqref{individual based _2groups}, we can first consider the dynamics on~$\Delta$ and then the dynamics transverse to~$\Delta$.

%%%
\subsubsection{Dynamics and bifurcations in symmetric systems}

To simplify notation, we will first consider the system
\begin{subequations}\label{eq:SymmetricPop}
\begin{align}
    \dot y_1 &= C_1y_1 + C_1'y_2 + C_2y_1^2 + C_2'y_1y_2 + C_3 y_1^3\\
    \dot y_2 &= C_1y_2 + C_1'y_1 + C_2y_2^2 + C_2'y_2y_1 + C_3 y_2^3.
\end{align}
\end{subequations}
with a polynomial vector field. Note that the SIS model~\eqref{individual based _2groups} is a special case with
$C_1 = \tau d_1-\gamma $,
$C_1' = \tau d_2 $,
$C_2 = -\tau d_1+\beta h_1 $,
$C_2' = -\tau d_2$, and
$C_3 = -\beta h_1$.
% \IZK{We've already said that this is dynamically invariant.}
% The symmetry implies that $\Delta = \{y_1=y_2=:y\}$ is dynamically invariant as a fixed point subspace of the $\Z_2$-action that contains the disease free equilibrium $y=0$.
Note that any point in~$\Delta$ is a \emph{symmetric state}, that is, the state of both populations take the same value.
Any point in $\R^2\smallsetminus\Delta$ we refer to as an \emph{asymmetric state}.

First, we will consider the dynamics of symmetric states within~$\Delta$; these behave like a single population (cf.~Section~\ref{sec:OnePopTriangles}).
Specifically, the dynamics on~$\Delta$ evaluate to
\begin{align}
\label{eq:DynDelta}
    \dot y &= (C_1+C_1')y + (C_2+C_2')y^2 + C_3 y^3 .
\end{align}
If $C_2+C_2'\neq 0$, the trivial equilibrium $y = 0$ undergoes a transcritical bifurcation at $C_1+C_1' = 0$ within~$\Delta$.
The emergent branch of nontrivial equilibria is determined by the quadratic equation
\begin{equation}
C_1+C_1' + (C_2+C_2')y + C_3 y^2 = 0.\label{nontriv_ss}
\end{equation}
Depending on the parameter values there may or may not be a fold point for $y>0$.

Second, we can determine dynamics and bifurcations close to these equilibria transverse to~$\Delta$ to get insights into the full dynamics of~\eqref{eq:SymmetricPop}.
Linear stability at any point $(y_1, y_2)= (y,y)\in\Delta$ is given by
\begin{equation}
    A^{\Delta} = \left(\begin{array}{cc}
        C_1+(2C_2+C_2')y+3C_3y^2&  C_1'+C_2'y\\
        C_1'+C_2'y & C_1+(2C_2+C_2')y+3C_3y^2
    \end{array}\right).
\end{equation}
Naturally, the symmetry constrains the eigenvalues and eigenvectors of~$A^{\Delta}$ given by
\begin{subequations}
\begin{align}
\label{eq:ParEV}
    \lambda^\parallel(y) &= C_1+(2C_2+C_2')y+3C_3y^2+C_1'+C_2'y, & v^\parallel &= (1,1)^T\\
    \lambda^\perp(y) &= C_1+(2C_2+C_2')y+3C_3y^2-C_1'-C_2'y, & v^\perp &= (1,-1)^T.
\label{eq:TransEV}    
\end{align}
\end{subequations}
The first eigenvalue corresponds to linear stability with respect to perturbations within~$\Delta$ that reduce to the one-dimensional system discussed above.
The second eigenvalue on the other hand corresponds to linear stability with respect to perturbations transverse to~$\Delta$.
Its zeros indicate bifurcations that give rise to nontrivial equilibria off~$\Delta$.
Due to the symmetry, such bifurcations are generically of pitchfork type.

We can compare the two eigenvalues to get insight into the relationship of bifurcations within and transverse to~$\Delta$.
Note that for $0\leq y \leq 1$, the sign of~$C_1'+C_2'y$ determines whether $\lambda^\parallel \leq \lambda^\perp$ or $\lambda^\parallel \geq \lambda^\perp$. 
Specifically, if $C_1'+C_2'y=\tau d_2 (1-y)\geq 0$, we have
$\lambda^\perp \leq \lambda^\parallel$ and the stability of any equilibrium within~$\Delta$ ($\lambda^\parallel<0$) implies stability in the full system~\eqref{eq:SymmetricPop}.
This also means that any stable equilibrium within~$\Delta$ first loses stability within~$\Delta$ (in a transcritical bifurcation) before it can lose stability in the transverse direction.

\begin{comment}
Moreover, 
\begin{align}
\lambda^\perp(y) &=0, & \lambda^\parallel(y) < 0
\end{align}
now gives an explicit condition for a pitchfork bifurcation that creates bistability between a pair of stable equilibria (one with $y_1<y_2$ and one with $y_1>y_2$) in~$\R^2$ that are related by symmetry. It was shown above that $\lambda^\perp \leq \lambda^\parallel$, hence there is no pitchfork bifurcation leading to a stable non-symmetric steady state. Nevertheless, we determine the bifurcation curve of the non-symmetric steady states below. It will turn out that despite of the fact that the non-symmetric steady state is unstable in a neighbourhood of the pitchfork point, it becomes stable for other values of $\tau$. 

\end{comment}

%%%
\subsubsection{Equilibria and bifurcations of symmetric states}

The first step to gain insights into the SIS dynamics~\eqref{individual based _2groups} is to analyze the dynamics and bifurcations within the invariant subspace~$\Delta$.
The trivial equilibrium undergoes a transcritical bifurcation at $0= \tau( d_1+d_2)-\gamma = C_1+C_1' = 0$ or
\[
\tau = \frac{\gamma}{d_1+d_2}
\]
if $C_2+C_2' = \beta h_1-\gamma \neq 0$.
Moreover, the sign of $\beta h_1-\gamma$ (the quadratic term in~\eqref{eq:DynDelta}) determines the type of transcritical bifurcation: 
It is of forward type when $\gamma > \beta h_1$ and it is of backward type for $\gamma < \beta h_1$.

The nontrivial branch of equilibria~\eqref{nontriv_ss} is given by
\begin{equation}
\tau = \frac{\gamma}{(d_1+d_2)(1-y)} - \frac{\beta h_1}{d_1+d_2}y . \label{tau_y_curve}
\end{equation}
We can see that $\tau \to \infty$ as $y\to 1$.
As the solution to a quadratic equation, the type of bifurcation of the trivial equilibrium also determines bistability: There is a fold point for $0<y<1$ if  $\gamma < \beta h_1$. 
In this case, bistability occurs, i.e., for certain values of $\tau$ the larger endemic steady state is stable together with the trivial steady state $y=0$. 
The following theorem summarizes the results about the bifurcation diagram of the symmetric model.

\begin{theorem}
	Transcritical bifurcation occurs in two symmetric populations~\eqref{eq:SymmetricPop} at $\tau = \gamma/(d_1+d_2)$. The shape of the bifurcation curve of non-trivial steady states depends on the parameters that determine the higher-order interactions as follows.
	\begin{itemize}
		\item If $ \beta h_1 < \gamma$, then the $(\tau, y)$ bifurcation diagram is of forward type and the non-trivial steady state is stable. 
		\item If $\gamma < \beta h_1$, then the $(\tau, y)$ bifurcation diagram is of backward type and it has a fold point. The non-trivial steady state on the lower branch (below the fold point) is unstable, while the non-trivial steady state on the upper branch (above the fold point) is stable.
	\end{itemize}
\end{theorem}

\subsubsection{Pitchfork bifurcation towards non-symmetric steady states}

In a second step, we now analyse potential bifurcation points of the nontrivial branch of symmetric equilibria~\eqref{tau_y_curve} within~$\Delta$ in the transverse direction.
For the SIS dynamics, the bifurcation condition for the general expression of the transverse eigenvalue~\eqref{eq:TransEV} evaluates to
\begin{equation}
   \tau(d_1-d_2)-\gamma+2(\beta h_1-\tau d_1)y-3\beta h_1 y^2 = 0.\label{pitchfork}
\end{equation}
Eliminating $\tau$ from~\eqref{tau_y_curve} and~\eqref{pitchfork} gives that the
$y$-value of the pitchfork point has to satisfy
\begin{equation}
f(y) := 2\gamma d_2 + (d_1-d_2)\gamma y -(d_1+3d_2)\beta h_1y (1-y)^2 = 0.\label{eqn_fy}
\end{equation}
The goal now is to find a solution $y\in [0,1]$ for which~\eqref{tau_y_curve} yields a positive value of $\tau$, i.e., $\gamma - \beta h_1 y (1-y) >0$. 

Depending on the values of the parameters, the function~$f$ may have zero or two roots in the interval $[0,1]$.
At the endpoints of the interval, we have $f(0)=2\gamma d_2 >0$ and $f(1)=\gamma (d_1 +d_2)>0$. 
Then the bifurcation conditions for the appearance of two roots is $f(y)=0$ and $f'(y)=0$, that is,
\begin{align}
\label{eq:FrstEq}
2\gamma d_2 + (d_1-d_2)\gamma y -(d_1+3d_2)\beta h_1y (1-y)^2& = 0,\\
(d_1-d_2)\gamma  -(d_1+3d_2)\beta h_1(1-3y)(1-y)& = 0 .
\label{eq:SndEq}
\end{align}
Dividing both equations by~$d_2$ we can observe that the single parameter $d=d_1/d_2$ determines the existence of the pitchfork bifurcation point instead of the two parameters~$d_1$ and~$d_2$.
Dividing the first equation~\eqref{eq:FrstEq} by the second~\eqref{eq:SndEq}, we can express the new parameter~$d$ in terms of~$y$ as
\[
\frac{d_1}{d_2} = d=\frac{1-3y+y^2}{y^2} .
\]
Then substituting this expression into the second equation~\eqref{eq:SndEq}, one can determine the other parameters as
\[
\beta h_1= \frac{\gamma}{(1-3y+4y^2)(1-y)} .
\]
It is easy to check that the bifurcation curve, parameterised by the above two equations is a monotone in the $(d,h_1)$ plane. 
Pitchfork bifurcation points, and hence symmetry breaking bifurcation, occur along this curve. 
% \IZK{I do not get the S notation.}
The implicit relation between~$d$ and~$h_1$ given by the parametrisation in terms of~$y$ can also be expressed as a direct relation between the two variables, namely  we can introduce a function $S$ such that $S(d)=h_1$.
Using this curve, we can complete the characterisation of the bifurcation diagrams in the $(\tau, y)$-plane as it is presented in the theorem below.

\begin{theorem} \label{theo:two_pop_bif}
The parameter plane $(d_1/d_2 , h_1)$ can be divided according to the shape of the $(\tau, y)$ bifurcation diagram as follows. 
\begin{itemize}
	\item If $ h_1 < \gamma/\beta$, then the $(\tau, y)$ bifurcation diagram is of forward type and there is no pitchfork point. 
	\item If $\gamma/\beta < h_1 < S(d_1/d_2)$, then the $(\tau, y)$ bifurcation diagram is of backward type and there is no pitchfork point. 
	\item If $ S(d_1/d_2) < h_1$, then  the $(\tau, y)$ bifurcation diagram is of backward type and there are two pitchfork points, and two branches of non-symmetric solutions.
\end{itemize}
\end{theorem}

%%%
\subsubsection{Nonsymmetric solutions of the symmetric equation}

Can there be bistability between two nonsymmetric endemic steady states for the SIS dynamics~\eqref{individual based _2groups}?
The pitchfork bifurcations discussed above give rise to nonsymmetric steady states, but these may not necessarily be stable close to the bifurcation point.

Indeed, the first result shows that equilibria emerging in the pitchfork bifurcations are unstable. 
Recall that for equilibria $y\in\Delta$ with $0\leq y < 1$ the eigenvalues of the linearization satisfy $\lambda^\perp \leq \lambda^\parallel$.
This means that the transverse bifurcation condition $\lambda^\perp = 0$ can only be satisfied if~$\lambda^\parallel > 0$.
In summary we have

\begin{proposition}
    For the SIS dynamics~\eqref{individual based _2groups}, any pitchfork bifurcation of the symmetric branch of equilibria~\eqref{tau_y_curve} within~$\Delta$ is subcritical. 
    That is, the emerging branch of nonsymmetric equilibria are unstable near the bifurcation point.
\end{proposition}

Thus, secondary bifurcations along the branch are necessary for bistability between nonsymmetric equilibria to emerge.
To get insight where these bifurcations may appear, we compute the branch of asymmetric equilibria in the $(\tau, y_1)$-plane. 
Specifically, equilibria of~\eqref{individual based _2groups} for symmetric coupling satisfy
\begin{subequations}  \label{individual based _2groups_symm_ss}
	\begin{align}
	\tau (d_1y_1 + d_2 y_2) &=  \frac{\gamma y_1}{1-y_1} - \beta h_1 y_1^2 ,\\
	\tau (d_2y_1 + d_1 y_2) &=  \frac{\gamma y_2}{1-y_2} - \beta h_1 y_2^2 .
	\end{align}
\end{subequations}
Dividing the first equation by the second one, we can eliminate~$\tau$ and obtain a relation between~$y_1$ and~$y_2$ as
\[
(d_1y_1 + d_2 y_2) \left( \frac{\gamma y_2}{1-y_2} - \beta h_1 y_2^2 \right) =  (d_2y_1 + d_1 y_2) \left( \frac{\gamma y_1}{1-y_1} - \beta h_1 y_1^2 \right) .
\]
This can be rearranged to
\[
\gamma \frac{d_2(y_1+y_2) + (d_1-d_2)y_1y_2}{(1-y_1)(1-y_2)} = \beta h_1 \left( (d_1-d_2)y_1y_2 + d_2 (y_1+y_2)^2 \right) .
\]
Dividing this equation by $\gamma d_2$ and writing $d=d_1/d_2$, $B= \beta h_1/\gamma$, $x=(y_1+y_2)/2$, and $z=y_1y_2$, yields
\[
\frac{2x + (d-1)z}{1-2x+z} = B ((d-1)z +4x^2).
\]

\begin{figure}%[h!]
    \centering
    \includegraphics[width=0.7\linewidth]{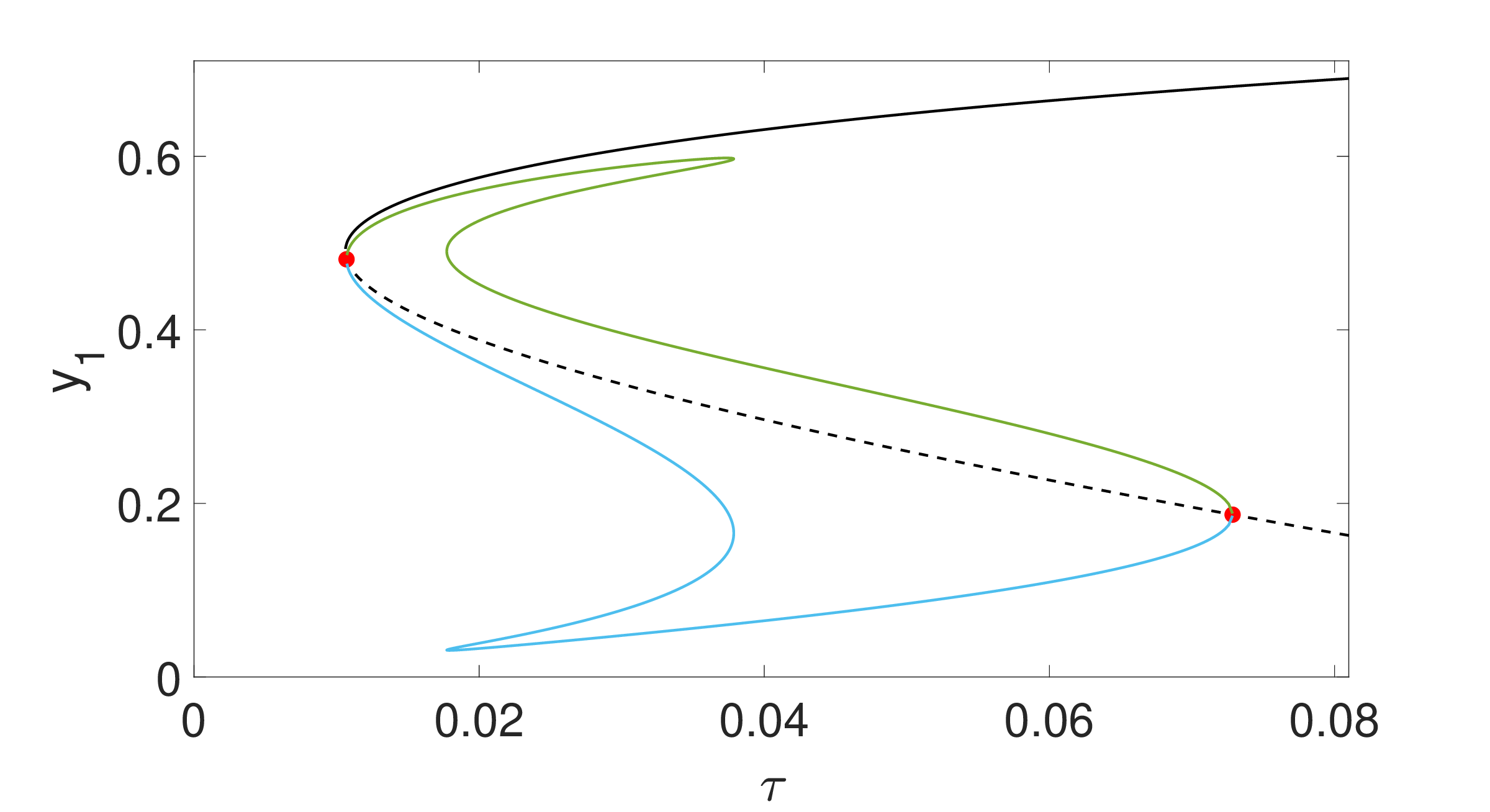}
    \caption{
    Branches of equilibria in the ($\tau, y_1$)-plane for two symmetric populations~\eqref{eq:SymmetricPop} with parameters $d_1=4$, $d_2=3$, $h=1$, $\beta=3.85$, and $\gamma=1$.
    The black line corresponds to the symmetric equilibrium branch in~$\Delta$; a solid line corresponds to a stable equilibrium, a dashed line to an unstable solution.
    Change of transverse stability occurs at pitchfork bifurcations indicated by red dots. 
    The emergent asymmetric steady states are depicted by green and blue curves (stability not shown).
    % \textcolor{red}{What does $y$ mean here?}\quad  \textcolor{green}{It is $y_1$ but I am yet to locate figure.}
    }
    \label{fig:two_pop_bif_curve}
\end{figure}

This now enables us to express~$y_2$ in terms of~$y_1$ and compute the~$(\tau, y_1)$ branch parametrised by~$x$. 
Specifically, choosing a value of $x$ (the proper interval where~$x$ is varied will be defined later), the value of~$z$ can be determined from the above equation which is a quadratic one in~$z$. 
Once~$x$ and~$z$ are given, the values of $y_1$ and $y_2$ are determined from the system below
\[
y_1+y_2 = 2x, \qquad y_1y_2 = z,
\]
as
\[
y_{1,2} = x \pm \sqrt{x^2-z} .
\]
It is easy to check that both roots are real and are in the interval $[0,1]$ if and only if
\[
2x-1< z < x^2 <1 \quad \mbox{and} \quad x>0 \quad \mbox{and} \quad z>0 
\]
hold. 
Hence~$x$ is varied in the interval $[1/2 , 1]$ and this is restricted further by the upper and lower bound conditions on~$z$. 

As these equations are typically not accessible analytically, we can compute the $(\tau, y_1)$ solution curve numerically for fixed parameters~$d_1$, $d_2$, $h$, $\beta$ and $\gamma$.
Here we show a concrete example rather than attempting a full characterization.
Figure~\ref{fig:two_pop_bif_curve} shows a situation corresponding to the third case in Theorem~\ref{theo:two_pop_bif}, when there are two pitchfork points. 
The non-symmetric steady states are computed by using the algorithm detailed above. 
Numerical computation shows that the branch of non-symmetric steady states connects the two pitchfork points, as it is shown in Figure~\ref{fig:two_pop_bif_curve}.
Moreover, it turns out that this branch might have an $S$-shape, that is for certain values of~$\tau$ we can have six non-symmetric steady states. 
Hence the total number of steady states is nine. 
Two symmetric steady states are stable, and one is unstable. 
Four of the non-symmetric ones are saddles and two of them are stable. 
The phase portrait is shown in the top panel of Figure~\ref{fig:symmetric_phase_time_evol}. 
As we can see, there are four stable steady states (blue dots), and the boundary between their basins of attraction are formed by the stable manifolds of the saddle points (shown with magenta dots).  
The time dependence is shown in the bottom panel of Figure~\ref{fig:symmetric_phase_time_evol}, where four solutions starting from the four different basins are plotted.
The solutions are denoted by the letter corresponding to the stable steady state to which they converge.

\begin{figure}%[h!]
    \centering
    \includegraphics[scale=0.3]{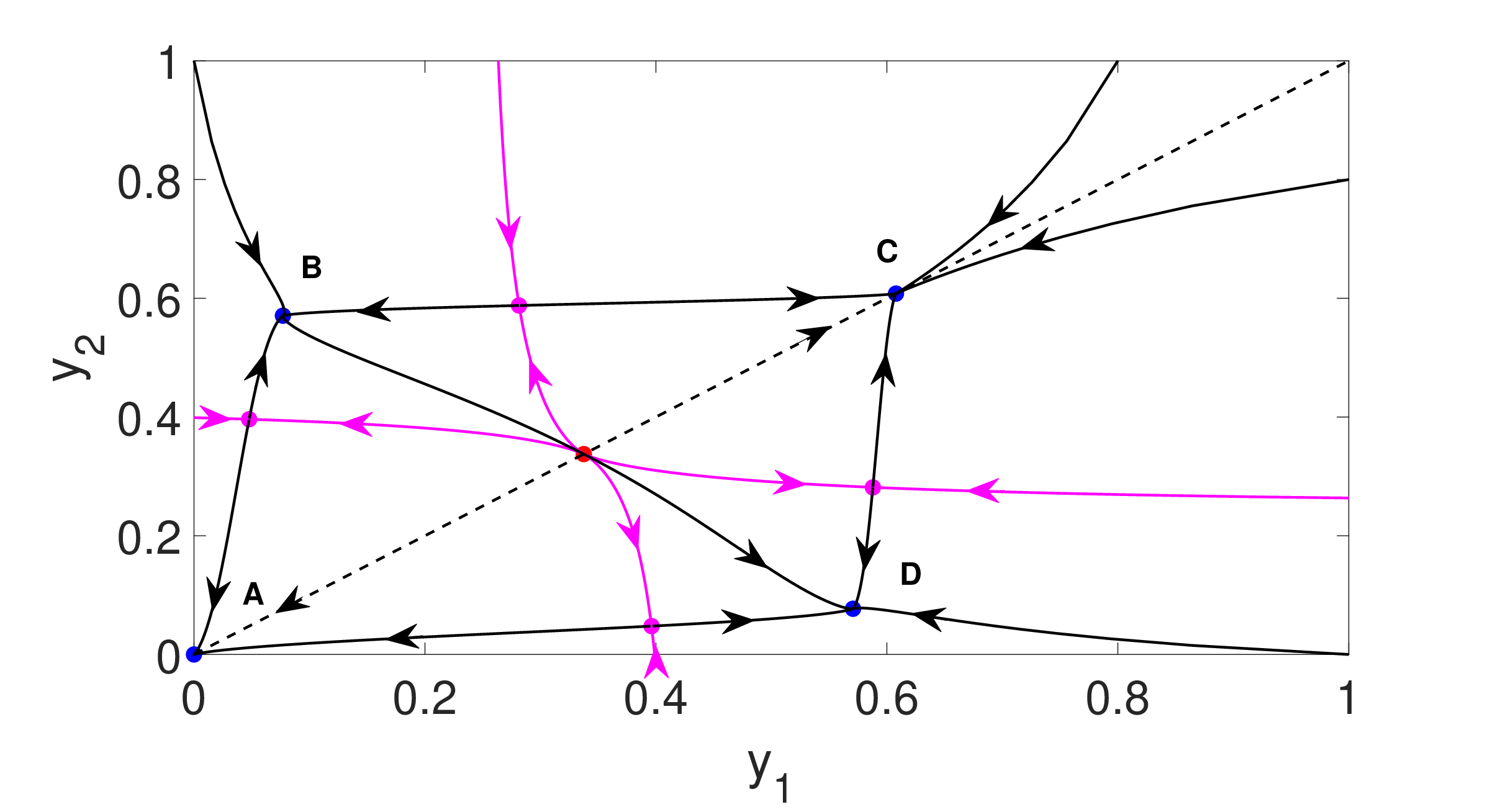}
\includegraphics[scale=0.25]{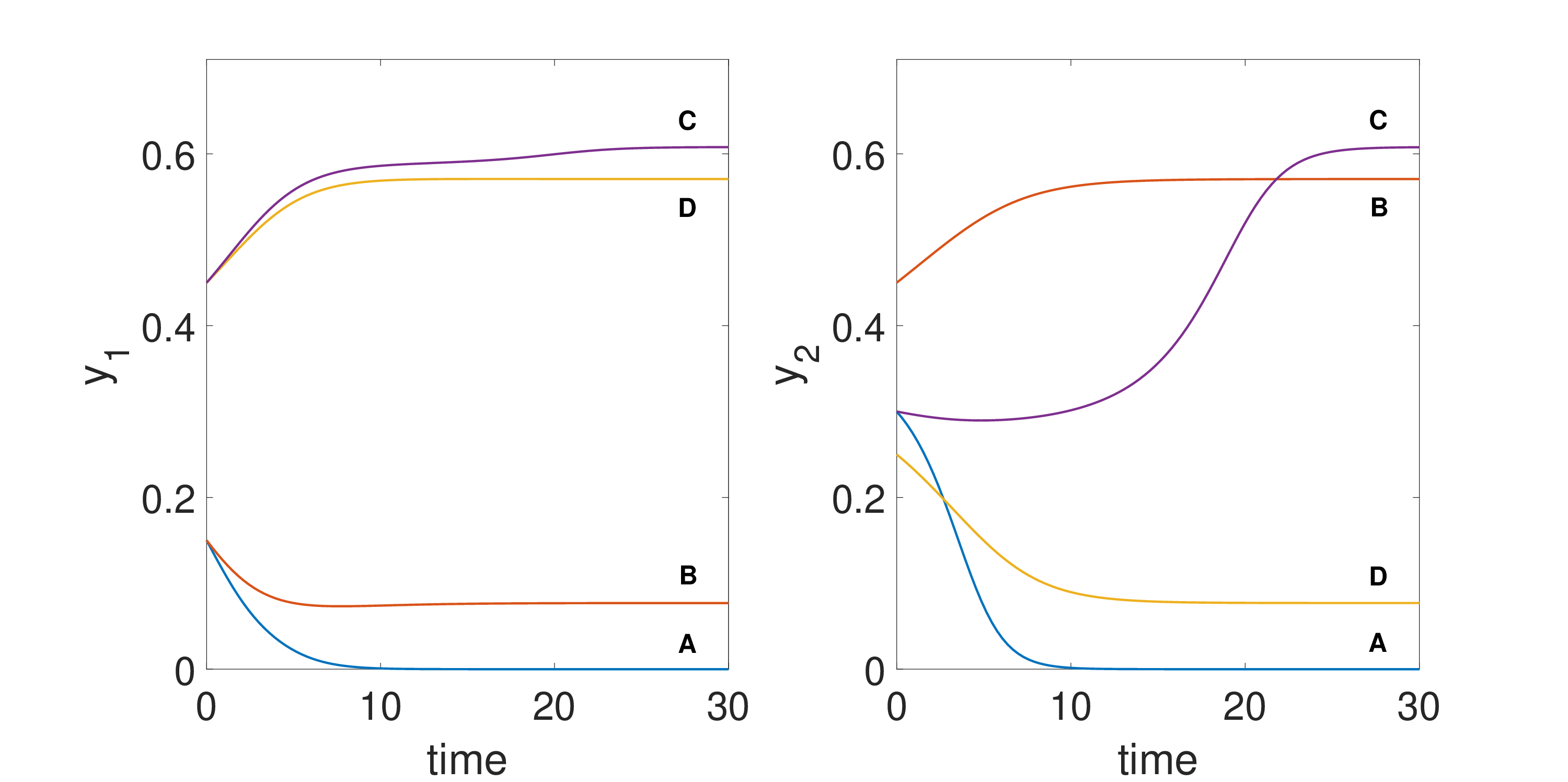}
    \caption{Illustration of the behaviour of the symmetric system~\eqref{eq:SymmetricPop} with parameters $d_1=4$, $d_2=3$, $\gamma=1$, $\beta=3.85$, $h=1$, and $\tau=0.03$. 
    Top: Phase plane with stable equilibria~A, B, C, D (blue dots), unstable symmetric equilibrium (red dot, and four additonal (unstable) saddle equilibria (magenta dots). 
    Bottom: Time evolution of~$y_1$ and~$y_2$ of solutions converging to one of the four labelled stable equilibria.}
\label{fig:symmetric_phase_time_evol}
\end{figure}

%%%%
\subsection{Two coupled populations with general coupling}

Now we return to the SIS dynamics~\eqref{individual based _2groups} for a general network without forcing symmetry.
Specifically, we consider perturbations away from symmetric coupling to gain insight into the dynamics and bifurcations.
This allows to exploit the fact that the conditions for bistability remain approximately true as the invariant subspace breaks. 
Indeed, the hyperbolic part of the branches are preserved while the pitchfork bifurcations may split into bifurcations that appear more generically (saddle-node/fold bifurcations).

As before, we can compute the steady states of the system. For~\eqref{individual based _2groups}, these satisfy
\begin{subequations}  \label{individual based _2groups_ss}
	\begin{align}
	\tau (d_1y_1 + d_2 y_2) &=  \frac{\gamma y_1}{1-y_1} - \beta h_1 y_1^2 ,\\
	\tau (d_3y_1 + d_4 y_2) &=  \frac{\gamma y_2}{1-y_2} - \beta h_2 y_2^2 .
	\end{align}
\end{subequations}
Dividing the first equation by the second one, we can eliminate~$\tau$ and obtain a relation between~$y_1$ and~$y_2$ as
\[
(d_1y_1 + d_2 y_2) \left( \frac{\gamma y_2}{1-y_2} - \beta h_2 y_2^2 \right) =  (d_3y_1 + d_4 y_2) \left( \frac{\gamma y_1}{1-y_1} - \beta h_1 y_1^2 \right) .
\]
This can be rearranged to
\begin{subequations}
\begin{align} \label{eq:y1y2}
&\gamma \left( d_1-d_4) y_1y_2 +d_2 y_2^2- d_3 y_1^2 + (d_3-d_1)y_1^2y_2+  (d_4-d_2) y_1y_2^2 \right) \\ &\qquad= 
\beta \left( h_2d_2y_2^3 -h_1d_3y_1^3 + h_2d_1y_1y_2^2 - h_1d_4y_1^2y_2  \right) (1-y_1)(1-y_2).  &
\end{align}
\end{subequations}
These equations can be solved numerically to determine branches of equilibria.
Specifically, solve the quartic equation in~$y_1$ numerically for a given value of~$y_2$. Then~\eqref{individual based _2groups_ss} can be used to compute~$\tau$. 
Hence the equilibrium curve in the $(\tau, y_1)$-plane can be obtained semi-analytically parametrised by~$y_2$ as follows. 
The value of~$y_2$ is varied in an appropriate interval and the values of~$y_1$ are obtained by solving numerically the quartic equation~\eqref{eq:y1y2}. 
The different solutions of this equation will lead to different branches of the bifurcation curve.

\begin{figure}%[h!]
    \centering
    \includegraphics[scale=0.3]{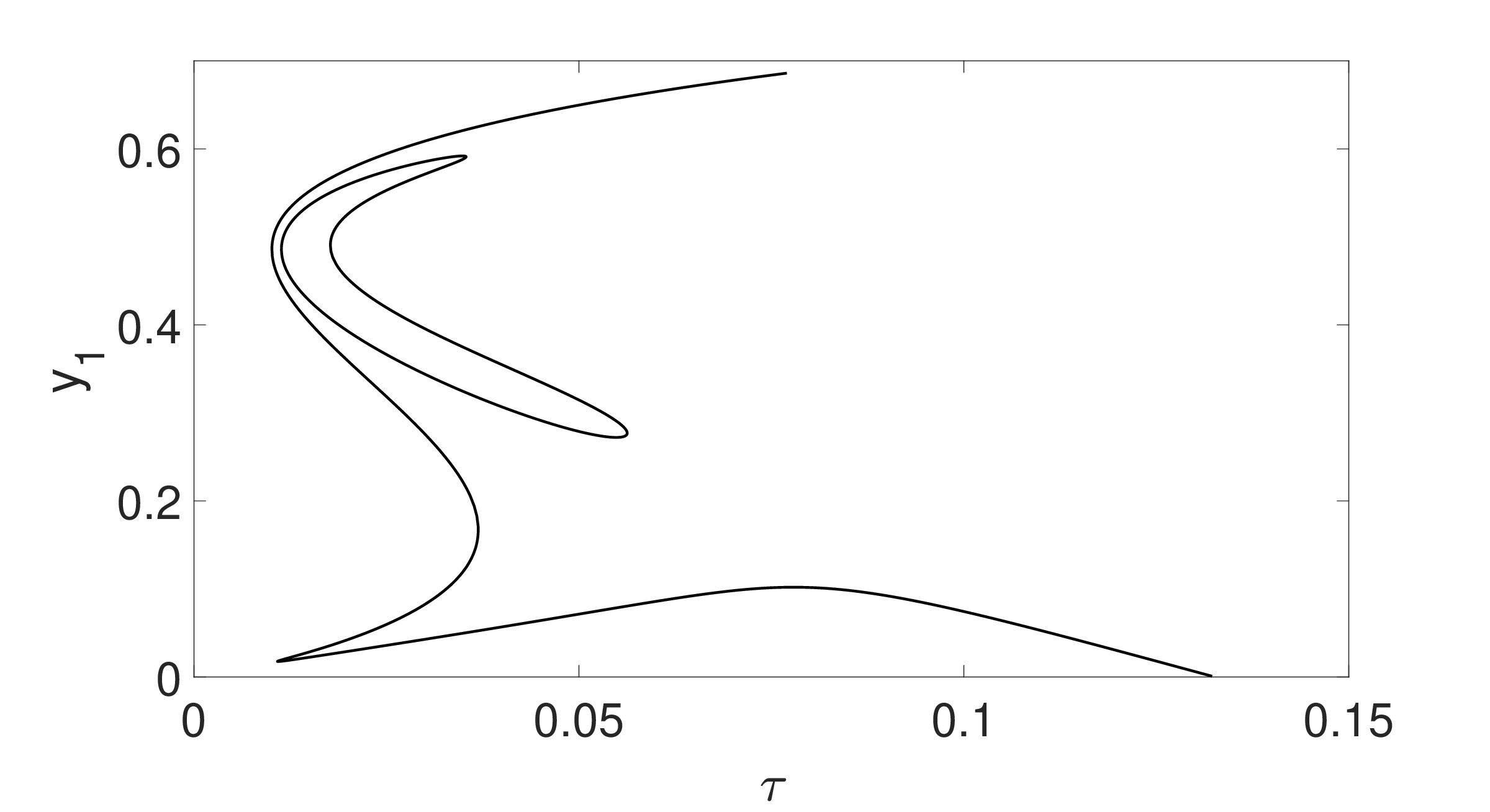}
\includegraphics[scale=0.3]{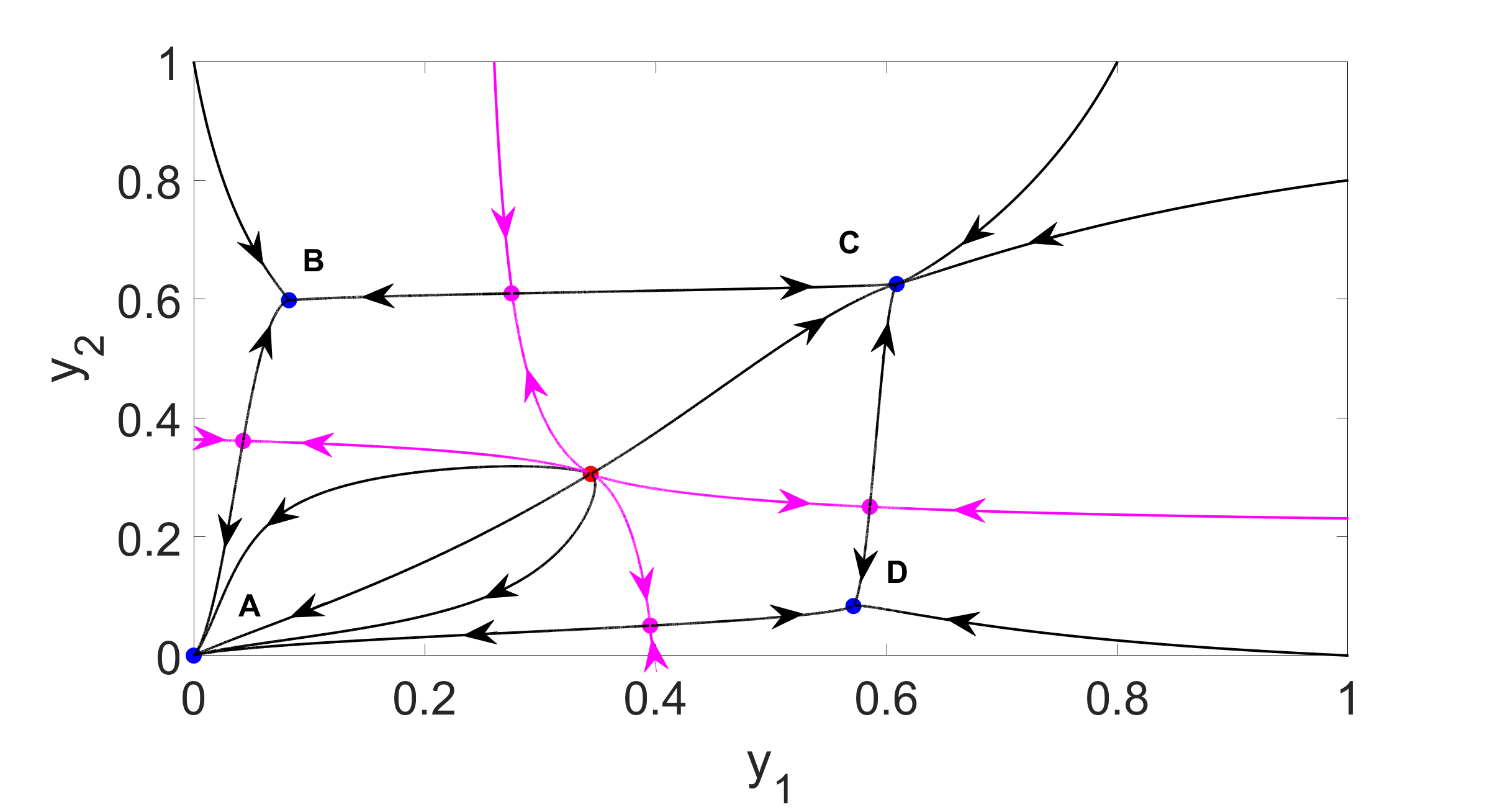}
    \caption{%
    Bifurcations and phase plane for two nonsymmetric populations~\eqref{individual based _2groups} with parameters $d_1=4$, $d_2=3$, $d_3=4$, $d_4=5$, $h_1=1$, $h_2=1.01$, $\gamma=1$, and $\beta=3.85$.
    Top panel: Branches of equilibria curve in the $(\tau, y_1)$ plane. The bifurcation curve consists of two parts, an isola and an unbounded part.
    Bottom panel: Phase plane for fixed $\tau=0.03$ shows the stable equilibria A, B, C, and~D (blue dots), an unstable symmetric equilibrium (red dot), and four additonal (unstable) saddle equilibria (magenta dots). The stable manifolds of the saddle points (magenta curves) separate the basins of attractions of the stable equilibria.}  
\label{fig:non_symmetric_bif_phaseplane}
\end{figure}

To illustrate the dynamics, we used parameter values close to those in Figure~\ref{fig:two_pop_bif_curve}; hence this case can be considered as a perturbation of the symmetric case. 
The bifurcation curve belonging to non-symmetric parameter values is shown in the top panel of Figure~\ref{fig:non_symmetric_bif_phaseplane}. 
The symmetry $d_1=d_4$ and $h_1=h_2$ of the parameter values is broken by choosing $d_1=4$, $d_4=5$ and $h_1=1$, $h_2=1.01$ that can be considered as a small perturbation of the symmetric case.
We can see that the pitchfork curves splits into more generic bifurcations. 
The number of fold points along the branches remains the same as in the symmetric case but an isola appears in the bifurcation diagram. 
One can see from the bifurcation diagram that there is a domain of~$\tau$ values for which there are 8 non-trivial steady states, three of them are stable, four are saddle points and one is unstable with two-dimensional unstable subspace. 
The phase portrait corresponding to $\tau=0.03$ is shown in the bottom panel of Figure~\ref{fig:non_symmetric_bif_phaseplane}. 
This phase portrait is qualitatively equivalent to that shown in the top panel of Figure~\ref{fig:symmetric_phase_time_evol}.

%%%%%%%%%%%%%%%%%%%%%%%%%%%%%%%%
\section{Discussion}\label{sec:discussion}
%%%%%%%%%%%%%%%%%%%%%%%%%%%%%%%%
% \IZK{Should we say more about the very exotic behaviour in the two pop model? The Discussion does not reflect this, at least omne more sentence?}
Reduced mean-field equations give valuable insights into disease spreading on networks that are challenging to obtain from high-dimensional state evolution.
Here we focused on a class of models called individual-based mean-field models which arise from considering a bottom-up approach, starting from the evolution equations for the probability of nodes being infected at time~$t$.
Breaking the dependency on higher-order moments early, results in an $N$-dimensional systems of differential equations which can be simplified to $M$~equations as long as nodes in any of the $M$~populations are topologically equivalent: The populations themselves can differ but equivalence is necessary within each individual population. Note that the insights we obtain go beyond fully homogeneous populations: 
Standard hyperbolicity considerations imply that one would expect similar dynamics if the homogeneity is broken (cf.~Section~\ref{sec:TwoPop}).
% The low-dimensional equations we consider are not specific  to the reduction approach that we employed but they rather  include other reduced model equations, that have been put forward in the literature, especially where downward closure is applied.

The low-dimensional equations we consider are not unique to the reduction approach we used; they also encompass other reduced models proposed in the literature, particularly those involving the application of downward closure.
This highlights that one expects to see global dynamical behavior that is qualitative similar across model equations.
% \IZK{Go beyond used to appear too many times, I fixed this above.}
The explicit expressions of how the model parameters relate to the network properties gives explicit insights into how network structure---and in particular the presence of higher-order interactions---affect the dynamics.
%Due to the similarity in the form of the equations in many different models and the qualitative similarity in global behaviour, we set out to analyse a class of such models from a mathematical viewpoint.
Interaction order has a direct influence on the possibility of multistability for one population:
We showed that models with up to three-body interactions can only give rise to either a transcritical transition or bistability and this is in line with findings in~\cite{kuehn2021universal}.
Models with four-body interactions however, show richer behaviour with the previous two possible outcomes being complemented by a multistability regime where two strictly non-zero endemic steady states can co-exist.
Furthermore, we investigated a coupled two-population model with up to three-body interactions only.
Here, we showed that symmetry breaking, a well studied phenomena in population dynamics, yields a possible route to multistability between endemic steady states. It is clear that the order of interaction affects the degree of the polynomial, with higher-order interactions leading to polynomials of higher degree. This in turn leads to richer behaviours, in our case up to four stable states, and this could present opportunities for further investigation.

%The exact modelling of complex contagion presents similar challenges to modelling epidemics on networks in presence of pairwise interactions alone. 
%In particular, the high-dimensionality of the model, scaling exponentially in the number of nodes in the network, prevents rigorous analysis or the modelling of the entire system at once; that is computing the probability of the system being in a certain state at a given time. 
%However, this is often not needed and stochastic simulations, such as Gillespie, allows us to compute exact stochastic paths of the full model. 
%While this is useful to gain intuition, it prevents us to gain analytical insight about possible global behaviour and to map out how contact structure and complex contagion interact and facilitate or hinder different outcomes.

Our insights serve as a valuable guide as to what qualitative behavior one may expect and where (in parameter space) it may arise.
Of course, these models only serve as an approximation of microscopic models such as individual-based stochastic models whose number of possible states scale exponentially in the number of individuals.
Indeed, given that individual-based mean-field models apply the closure early; at the level of pairs, it means that it misses out some of the important local correlations.
%Broadly speaking, this means that for certain networks with specific structures and density of edges, 3-body interactions etc, the agreement may be satisfactory to good. 
Standard results that relate exact and approximate models are typically over finite time (rather than asymptotics); see, for example,~\cite{wormald1995differential,wormald1999differential,sclosa2024}.
One typically expects that good agreement is possible away from bifurcation points.
In fact, direct simulations show that the agreement between  exact and approximate model can be surprisingly good.
While our mean-field models are thus instructive to understand how network structure shapes the contagion dynamics, a direct comparison with stochastic simulations are beyond the scope of this article.

Complex contagion through higher-order interactions shape the dynamics and make it richer compared to classical transmission models. In this paper, we provide a first step towards a mathematical classification of possible dynamics as well as an understanding of the more subtle interactions between structure and dynamics. However, there is further scope to extend this work to addressing questions around overlap between higher-order structures~\cite{malizia2023hyperedge,malizia2023pair,burgio_triadic_2024} and  adaptivity~\cite{burgio2023adaptive} as well as extension to variations of SIS model by including multiple disease classes leading to models such as SEIS, SIRS or even contact tracing. Note that such extensions should be driven by a well-defined underlying research question, rather than the mere application of existing tools to marginally adjusted models. Finally, we believe that our approach and results underscore the utility of bifurcation theory as a powerful tool for elucidating potential system dynamics and enhancing the understanding of how and why certain outcomes are either promoted or constrained.

\bibliography{project1.bib}
\bibliographystyle{abbrv}

%%%%%%%%%%%%%%%%%%%%%%%%%%%%%%%%%%%%%%%%%%%%%%%%%%%%%%%%%%%%%%%%

\subsection*{Data statement}
This is a theoretical study and no data was used. Code will be made available upon a reasonable request.

\subsection*{Competing interests}
We declare that we have no competing interests.

\subsection*{Acknowledgment}
 P.L. Simon acknowledges support from the Hungarian Scientific Research Fund, OTKA Grant No. 135241 and from ERC Synergy Grant No. 810115 - DYNASNET.

\end{document}